\documentclass{article}

\usepackage{natbib}

\usepackage[ruled,vlined, lined, boxed]{algorithm2e}
\renewcommand{\algorithmcfname}{ALGORITHM}
\SetAlFnt{\small}
\SetAlCapFnt{\small}
\SetAlCapNameFnt{\small}
\SetAlCapHSkip{0pt}
\IncMargin{-\parindent}

\usepackage[letterpaper, width=6.8in, height=9.3in, marginratio={1:1, 1:1}]{geometry}

\usepackage{fullpage}
\usepackage{float}
\usepackage{times}
\usepackage{epsfig}
\usepackage{amsmath}
\usepackage{amssymb}
\usepackage{amstext}
\usepackage{xspace}
\usepackage{latexsym}
\usepackage{verbatim}
\usepackage{multirow}
\usepackage{graphicx}
\usepackage{epstopdf} 
\usepackage{ifthen}
\usepackage{subcaption}
\usepackage{enumitem}
\usepackage{balance}

\usepackage{enumitem}

\usepackage[dvipsnames,usenames]{xcolor}
\usepackage{prettyref}
\usepackage[hypertexnames=false,colorlinks=true,pdfpagemode=Usenone,urlcolor=Blue,linkcolor=RoyalBlue,citecolor=OliveGreen,pdfstartview=FitH]{hyperref}

\newcommand{\AutoAdjust}[3]{\mathchoice{ \left #1 #2  \right #3}{#1 #2 #3}{#1 #2 #3}{#1 #2 #3} }

\newcommand{\InBrackets}[1]{\AutoAdjust{[}{#1}{]}}
\newcommand{\Ex}[2][]{\operatorname{\mathbf E}_{#1}\InBrackets{#2}}

\newcommand{\fall}[1]{\underbar{r}}

\newcommand{\argmax}{\operatorname{argmax}}

\newcommand{\rounds}{n}
\newcommand{\stage}{k}

\newcommand{\val}{v}

\newcommand{\dist}{F}

\newcommand{\distlh}{{\dist_{[\low,\high]}}}
\newcommand{\distlt}{{\dist_{[\low,t]}}}

\newcommand{\distab}[1][a,b]{{\dist_{[#1]}}}

\newcommand{\ps}{p^{\ast}}
\newcommand{\pslh}{{\ps_{[\low,\high]}}}

\newcommand{\psab}[1][a,b]{{\ps_{[#1]}}}

\newcommand{\Rlh}{{R_{[\low,\high]}}}

\newcommand{\Rab}[1][a,b]{{R_{[#1]}}}

\newcommand{\prej}{p_{20}}
\newcommand{\pacc}{p_{21}}
\newcommand{\prejrej}{p_{300}}
\newcommand{\prejacc}{p_{301}}

\newcommand{\low}{\ell}
\newcommand{\high}{\mathfrak{h}}

\newcommand{\bel}{\mu}

\newcommand{\mstrat}{\sigma}
\newcommand{\mstrats}{\mstrat_s}
\newcommand{\mstratb}{\mstrat_b}

\newcommand{\hist}{h}

\newtheorem{theorem}{Theorem}
\newtheorem{definition}{Definition}

\newtheorem{lemma}{Lemma}

\newtheorem{fact}{Fact}
\newtheorem{proposition}{Proposition}

\newtheorem{remark}{Remark}
\newcommand{\qed}{\mbox{\ \ \ }\rule{6pt}{7pt} \bigskip}

\renewcommand{\comment}[1]{}
\newenvironment{proof}{\noindent{\em Proof:}}{\hfill\qed}

\makeatletter
\newenvironment{oneshot}[1]{\@begintheorem{#1}{\unskip}}{\@endtheorem}
\makeatother

\begin{document}
\title{Perfect Bayesian Equilibria in Repeated Sales
\thanks{Part of this work was done while the third author was at Microsoft 
Research. A preliminary version of this work (\cite{DPSSODA15}) appeared in 
the 
Proceedings of the ACM-SIAM Symposium on Discrete Algorithms, 2015.}}
\author{Nikhil R. Devanur\thanks{Microsoft Research. \tt{nikdev@microsoft.com}.}
 \and Yuval Peres\thanks{Microsoft Research. \tt{peres@microsoft.com}.}
 \and Balasubramanian Sivan\thanks{Google Research. \tt{balusivan@google.com}.}
}

\date{}
\maketitle{}
\begin{abstract}
A special case of Myerson's classic result  describes the revenue-optimal 
equilibrium when a seller offers a single item to a buyer. We study a 
{\em repeated sales} extension of this model: a seller offers to sell a single 
{\em fresh copy} of an item to the {\em same buyer every day} via a posted 
price. The buyer's private value for the item is drawn 
initially from a publicly known distribution $F$ and remains the same 
throughout. A key aspect of this game is that the seller might try to learn the buyer’s private value to extract more revenue, while the buyer is motivated to hide it. We study the Perfect Bayesian Equilibria (PBE) in this setting with varying levels of commitment power to 
the seller. We find that the seller having the commitment power to not raise 
prices subsequent to a purchase significantly improves revenue in a PBE.
\end{abstract}

\section{Introduction}
\label{sec:intro}
Most interesting economic games are inherently dynamic and/or repetitive, with
the same sellers repeatedly interacting with the same buyers.  Such scenarios
arise commonly in e-commerce platforms, such as eBay and Amazon, and online
advertising markets, such as Google, Yahoo! and Microsoft, among others. 
Unfortunately, the game-theoretic aspects of such repeated interactions are
poorly understood compared to their static, one-shot counterparts. In this
paper, we develop the theory for one such fundamental setting.

\paragraph{Questions we study.} 
There is a single seller of a certain good (say fish) and a single buyer who enjoys consuming a fresh fish 
every day. The buyer has a private value $v$ for each day's fish, drawn from a publicly known 
distribution. However, this value is drawn {\em only once}, i.e., the buyer has the same
unknown value on all days.  Each day, the seller sets a price for that day's
fish, which of course can depend on what happened on previous days. The buyer
can then decide whether to buy a fish at that price or to reject. The goal  of
the buyer is to maximize his total utility (his value minus price on each day
he buys and 0 on other days), and the goal of the seller is to maximize
profit.  How much money can the seller make in $n$ days in a \emph{Perfect Bayesian Equilibrium 
	(PBE)}? The key point here is that the seller is unable to credibly commit to prices for the future days (we 
refer the reader to Appendix~\ref{app:intro_details} for a gentle introduction to a Perfect Bayesian 
Equulibrium, and the role of commitment in it). We study three different versions of this problem for 
arbitrary distributions: the $n$ rounds version without any commitment, the $n$ rounds version with 
partial commitment, and the
time discounted infinite horizon version with partial commitment. Here are the formal definitions.

\begin{definition} \label{def:basicgame} A 1 seller, 1 buyer {\em Finite
		Horizon Repeated Sales game} is a sequential (extensive form) game between a
	{\em seller} and a {\em buyer}, with $\rounds$ rounds.  In each round, the
	buyer has a private valuation of $\val$ for a {\em perishable} item,  with a
	quasilinear utility.  The value $\val$ is initially drawn from a distribution
	$\dist$ supported on $[\low,\high]$ ($0 \leq \low \leq \high$) and stays the
	same throughout;  the seller only knows $\dist$.  The seller can produce a
	fresh copy of the item in each round, at a publicly known cost (normalized to)
	0.  Each round has two stages: the seller first offers a price for the item,
	and then the buyer responds with an {\em accept} or a {\em reject}.
\end{definition} 

\begin{definition}
	A 1 seller, 1 buyer {\em Finite Horizon Partial Commitment Repeated Sales game}
	is the same as a finite horizon repeated sales game with the additional
	condition that the seller cannot raise prices once a purchase has been made at 
	a certain price. He still holds the
	freedom to lower prices. 
\end{definition}

\begin{definition}
	A 1 seller, 1 buyer {\em Time Discounted Infinite Horizon Partial Commitment Repeated Sales game} with a discount factor $1-\delta$ is a partial commitment repeated sales game that is played forever, with the buyer and the seller discounting their round $i$'s utility by $(1-\delta)^{i-1}$. Equivalently, it
	can be thought of as a partial commitment repeated sales game (without any discounting) whose
	stopping time is a geometrically distributed random variable: the probability that the game stops after any given round is $\delta$. 
\end{definition}

\paragraph{Interpretation as a game with geometric stopping time.} While the
infinite horizon might appear as a mathematical curiosity with little practical
relevance, it is actually the most realistic of the three models. The time
discounted infinite horizon game is exactly equivalent to a game with geometric
stopping time. With a discount factor of $1-\delta$ per round, the $i$-th round
utility is discounted by a factor $(1-\delta)^{i-1}$.  Equivalently, if the
game stops after any given round with probability $\delta$, the probability
that the $i$-th round is reached is $(1-\delta)^{i-1}$, and therefore, any
utility obtained in that round has to be discounted by a factor
$(1-\delta)^{i-1}$. Often, a geometric stopping time is more realistic than a
fixed $n$ day horizon because the buyer or the seller may not be sure of the
precise number of interactions that will take place. 

\paragraph{A revenue upper bound.} 
We already know that if the seller were able to fully commit to all future
prices, he can commit to setting Myerson optimal price for every single day,
thereby getting $n$ times Myerson optimal revenue.  By not committing to
future prices, can the seller get more revenue, or at least as much revenue?
Surprisingly, the seller cannot get any more revenue by not committing --- this 
was shown in~\cite{BB84}. The gist of the 
argument is that if
it were possible to extract more than $n$ times Myerson revenue, then
it would be possible to extract more than Myerson optimal revenue in a single
round game. We present this small and crisp argument in 
Proposition~\ref{thm:MyeOPT} for a very 
general mechanism design 
setting, with arbitrary objectives, arbitrary time discounting, with the number 
of repetitions of the game possibly being a random variable. We refer the 
reader to Appendix~\ref{app:Myerson_optimal} for a formal
description of the model and the proof of the proposition. 

\begin{proposition} (A simple generalization of the result in~\cite{BB84})
	\label{thm:MyeOPT}
	In the general model of repeated mechanism design, the optimal objective value 
	obtained without any commitment is never larger than the optimal objective
	value obtained when commitment is possible. 
\end{proposition}

\paragraph{The power of commitment and a revenue benchmark.} The goal of this 
paper is to explore the 
``power of commitment'' from the seller's point of view: to what extent does 
varying levels of commitment impact seller revenue? 
Proposition~\ref{thm:MyeOPT} says that the revenue obtained by a seller with 
full commitment power cannot be exceeded by a seller without commitment power, 
suggesting it as a benchmark for studying the 
power of commitment.  In particular, we seek to study the ratio of seller 
revenue obtained from 
a certain level of commitment to that of revenue from full commitment. This 
ratio quantifies the extent to commitment impacts 
seller revenue. 

The finite horizon repeated sales game (the first version) has been  previously 
considered by
\citet{HT88} and \citet{Sch93}. (See also the survey by \citet{FV06}.)
\citet{HT88} consider the special case where $\dist$ is a 2-point distribution
and \citet{Sch93} generalizes it to any discrete distribution. These papers 
show that for the finite 
horizon 
version of the game, a PBE always
exists, and, every PBE charges the minimum possible price $\low$ on all but the
final few constant number of rounds. For instance, for such a PBE in the 
$n$-rounds repeated sales problem where the buyer's value is $U[0,1]$, the 
seller extracts
only a constant amount of revenue, as opposed to our full-commitment-benchmark 
of $n/4$, namely, $n$ times the
Myerson 1-round revenue of $1/4$. While the benchmark grows linearly with $n$, 
the no 
commitment revenue does not even grow with 
$n$. The question we ask is whether \emph{partial} commitment from the seller 
not to raise prices can significantly improve the situation? A commitment not 
to \emph{raise prices} is observed in several settings: for example annual and 
two-year contracts often offered by Internet and telephone service providers 
can be seen as a guarantee not to raise prices in the future --- they could 
always offer a lower price if that would entice the buyer to purchase. On the 
other hand, a 
guarantee not to \emph{lower prices} is rarely seen in practice, and is also 
much harder to enforce. After all, it may be difficult
for the seller to resist the temptation to lower prices if it will entice the
buyer to purchase. Moreover while increasing the price may be beneficial to the
seller, it is never in the interest of the buyer, but decreasing the price
(when it is higher than the buyer's value) benefits both the buyer and the
seller.

\paragraph{Our results.}

The main message of this paper is that the the seller possessing the power 
to credibly commit not to raise prices can guarantee significantly higher 
revenue, than when he is unable to commit to anything, in very simple PBEs that 
we call threshold PBEs. We prove this via our analysis of settings 2 and 3 
below. For the sake of completeness, we also analyze the well-studied setting 1 
for the existence of threshold PBEs. Our main technical results are (a) the derivation of seller's PBE revenue, and the PBE prices posted for the linear demand case (i.e., 
for the $U[0,1]$ distribution where the demand diminishes linearly with price) in the finite horizon partial commitment game (setting 2 below), and (b) the derivation of equivalence (and its interesting consequences) between infinite horizon partial commitment game (setting 3 below) and the durable goods monopoly / bargaining with one-sided offer setting.
\paragraph{1. Finite horizon with no commitment.} PBEs being complex objects, 
we look for the existence of simple PBEs. One notion of simplicity in PBEs is 
that buyers follow threshold strategies: on each day, the buyer purchases only 
if his value is above a certain threshold (e.g. see~\cite{FT91},~\cite{FV06}). 
Such a pure strategy threshold PBE 
has, among other things, a very simple representation: the strategy and 
the Bayesian updated seller beliefs are simple to represent. We show that such 
pure strategy threshold PBEs do not exist in the finite horizon setting with no 
commitment for any atomless bounded support distribution (supported in $[0, 
\high]$). See Theorem~\ref{thm:nRounds}. 

\newcommand{\thr}{threshold} 

\paragraph{2. Finite horizon with partial commitment.} 
What if the seller has the power to credibly commit not to raise prices upon 
purchase? Please see ``The power of commitment'' discussion in the introduction 
for the motivation for such one-sided commitment. In our second result, we show that with partial commitment from the seller not 
to raise prices upon purchase, \thr\ PBEs are guaranteed to exist for all 
atomless bounded support distributions 
(Theorem~\ref{thm:partialCommitmentPBEExistence}).  For the case where the 
buyer value distribution $\dist$ 
is $U[0,1]$, the
seller's revenue is $\sqrt{\frac{n}{2} + \frac{\log n}{8} + O(1)}$, with a
horizon of $n$ rounds.  (See Theorem~\ref{thm:recursion}.) This is much 
better than the $O(1)$ revenue that we get when there is no commitment.  

\paragraph{3. Time discounted infinite horizon with partial commitment.} 
In our final result we consider the {\em infinite horizon game with time
	discounting}, combined with the power of partial commitment: the game is
repeated forever but time discounting ensures that players' utilities are still
finite, and the seller promises never to raise prices upon purchase. 
\begin{enumerate}[label=\alph*.]
	\item We establish a close connection between the partial commitment repeated 
	sales game and the literature on bargaining with one-sided offer and durable 
	goods monopoly in Theorem~\ref{thm:BargainingRepeated} (see related work 
	Section~\ref{sec:related} for a 
	definition of bargaining and durable goods monopoly settings). We show that for 
	every atomless bounded distribution: for every 
	PBE in a bargaining game (or equivalently the durable goods monopoly setting), 
	there is a corresponding PBE in the repeated sales 
	game with identical utility structure s.t. buyer utility and seller revenue in 
	the 
	repeated sales game are a factor $\frac{1}{\delta}$ larger than in the 
	bargaining game.
	
	\item We use this connection to confirm the famous Coase 
	conjecture (\citet{Coase1972}) in our setting for the linear demand case (i.e., 
	for the $U[0,1]$ distribution) 
	in Theorem~\ref{thm:eq_select}. In more detail: prior 
	work on bargaining and 
	durable goods monopoly
	(\citet{ST83},~\citet{Stokey81}) has 
	established that when $\delta \to 0$, the seller's first round price approaches 
	$0$, and almost the entire mass of buyers accept in the first round itself in 
	all PBEs where the buyer follows a stationary threshold strategy and seller 
	follows a scale-invariant strategy (see Section~\ref{sec:linearDemand} for 
	definitions of stationarity and scale-invariance). We arrive at the same 
	conclusions for our 
	partial commitment 
	repeated sales setting based on our reduction in 
	Theorem~\ref{thm:BargainingRepeated}. We note that although the revenue from 
	Theorem~\ref{thm:eq_select} confirms the Coase conjecture, the revenue of 
	$O(\frac{1}{\sqrt{\delta}})$ obtained from this PBE is quite high, and grows 
	with $\delta$, unlike the revenue from the PBEs in the finite horizon with no 
	commitment case. 
	
	\item When no equilibrium selection 
	is performed (i.e., not focusing on stationary strategies for buyers 
	etc.),~\citet{AD89} prove a 
	folk theorem that \emph{any} revenue between $0$ and Myerson optimal revenue 
	benchmark can be obtained in a PBE in the infinite horizon bargaining/durable 
	goods monopoly 
	setting. Again, our Theorem~\ref{thm:BargainingRepeated} immediately implies 
	that the folk theorem applies to our infinite horizon repeated sales setting 
	(Theorem~\ref{thm:folk}).
	
	\item For the linear demand case (i.e., $U[0,1]$ distribution), 
	\citet{ST83},~\citet{Stokey81} focus on the PBEs where the seller follows 
	scale-invariant strategies. What if we exogenously enforce that the seller 
	follows scale-invariant strategies? I.e., seller's strategy 
	space is restricted to scale-invariant strategies that post a price of $pk$ 
	when his belief is $U[0,k]$ for all $k$, with $p$ being independent of $k$. 
	With this exogenous enforcement, it turns out there is a unique PBE in which 
	the seller can extract at least  $\frac{4}{3+2\sqrt{2}} \sim 69\%$ of Myerson 
	optimal revenue benchmark, unlike the $O(\frac{1}{\sqrt{\delta}})$ predicted by 
	the Coase-conjecture-confirming Theorem~\ref{thm:eq_select} 
	(see Theorem~\ref{thm:infinitepartial} and the following 
	discussion for this result). Further, the price in the first round as $\delta 
	\to 0$ is about 
	0.586, unlike the $\sqrt{\delta}$ as the first round price predicted by 
	Theorem~\ref{thm:eq_select}. 
\end{enumerate}

\subsection{Related Work}
\subsubsection{Closely related work}
\label{sec:related} 
We discuss closely related work here.  See 
Appendix~\ref{app:broader_related_work} for a discussion of broader related 
work. 
\paragraph{Bargaining with one-sided offer.} In bargaining with
one-sided offer, a single seller repeatedly makes price offers
to a single buyer for the sale of a single unit of good, till a sale is made.
The buyer's private value $v$ for the good is drawn from distribution $\dist$
and is publicly known. In each round, the seller posts a price, and the buyer
can either accept or reject the offer the seller made. Once the buyer accepts at
a price, the game ends, and the seller's revenue is that round's price. We
are primarily interested in the infinite horizon bargaining game: the number of 
rounds in the game is unbounded, but the buyer and seller have a common 
discount rate of $1-\delta$
on their utilities, i.e., utilities in round $i$ are scaled by
$(1-\delta)^{i-1}$. Equivalently, one could think of the there being a
$1-\delta$ probability of the game ending after each round.

\paragraph{Durable goods monopoly.} The only difference of a durable goods
monopoly from the bargaining with one-sided offer setting is that instead of a 
single buyer with a
continuous type space, durable goods monopoly has a continuum of infinitesimal
buyers. But apart from that, just like in bargaining, there is a single seller,
and each buyer is interested in consuming exactly one good: once a buyer makes a
purchase the same buyer never purchases again because the good is durable and
lasts forever. The math for bargaining and durable goods monopoly are identical.

\paragraph{Bargaining, durable goods monopoly and Coase conjecture.} In his
analysis of the durable goods monopoly setting,~\cite{Coase1972} discussed
several properties of the equilibria. These include the threshold behavior from
buyers (resulting in higher value types buying earlier), the equilibrium path
exhibiting a decreasing sequence of prices till the buyer accepts to purchase
etc. Perhaps the most dramatic of these is that in a PBE the monopolist's 
profit (which
is same as his revenue since we assumed fixed marginal cost of $0$)
tends to zero when the discount factor $1-\delta$ tends to $1$. This
is surprising because a monopolist, who is by definition without competition,
should expect to extract a non-trivial amount of surplus as profit. The
reasoning is that the monopolist experiences competition from his own future
offers at a lower price --- in particular, the discount factor $1-\delta$
tending to $1$ implies that the seller makes his future offers at very quick
succession. I.e., he is unable to commit that there won't be any reselling in 
the 
future at lower prices. 

\citet{Bulow82} analyzed Coase's conjecture in a finite horizon model and showed
that the monopolist's price in every round (except of course the last round) is
indeed strictly smaller than the one-shot monopolist's price.~\citet{Stokey81}
further verified Coase's conjecture by studying the infinite horizon durable
goods monopolist game, and constructed an equilibrium that is the limit of
unique equilibria in the finite horizon models.~\citet{Gul1986} proved that 
there is a continuum of PBEs for this setting, but that all of them being 
qualitatively equivalent, and confirming Coase's conjecture in that the initial 
price of the seller converges to $0$.~\citet{ST83} analyzed the
infinite horizon model for bargaining with one-sided offer and obtained results
much along the lines of~\citet{Bulow82} and~\citet{Stokey81}, confirming 
Coase's conjecture. The PBE obtained is one where the seller follows a 
scale-invariant strategy, and the buyer follows a threshold strategy 
with a stationary threshold (explained in the proof). Both~\cite{FLT85} 
and~\cite{Gul1986} confirm the Coase conjecture in the infinite horizon 
bargaining game with one-sided offer in the ``gap'' case, where the smallest 
point in the buyer's value distribution is strictly larger than the seller's 
cost. 

\paragraph{Behavior based price discrimination.} 
The particular model of repeated sales we study has been investigated in the 
economics
literature, under the name {\em behavior based price discrimination} (BBPD)
\citep{FV06}.  The motivation there is that firms can offer personalized prices
to consumers based on their past consumption pattern. Such consumption patterns
could be collected in various ways, such as when the consumers use loyalty
cards, or in an online world where the consumer identifies himself by logging
in, or by the use of technology such as cookies.  \citet{FV06} give several
other markets where BBPD is observed,  such as magazine subscriptions and labor
markets.  BBPD is also prevalent in government and corporate  procurement, from
raw materials to IT infrastructure. 

The most closely related work to ours is that of \citet{HT88} and
\citet{Sch93}, and we have already discussed how our work relates to theirs.
Subsequently, many extensions of their models have been studied, such as when
consumer preferences vary over time, a monopolist seller selling multiple
goods, multile sellers selling the same good who try to poach customers from
each other, sellers with multiple versions of the same product, and so on. We
refer the reader to \citet{FV06} for a survey of these results.  Another
closely related paper is that of \citet{CTW12} who consider a repeated sale
game where the buyers have the option of anonymizing themselves at a cost and
analyze the effect of varying this cost on the welfare of the 
buyers.~\cite{KN14} consider a repeated sale setting where in each round two 
item types are auctioned, to a finite number of agents, with an information 
structure where agents know their own valuation but not the other agents' 
valuations. The seller is forced to run a second price auction, but can change 
the reserve dynamically.~\cite{KN14} show the optimality of static reserve 
under some assumptions and design an optimal dynamic mechanism when the 
assumptions fail. 

\paragraph{Follow-up work.}
\cite{ILPT17} study the multiple buyers ($n$ buyers) version of our problem, where there is a single fresh copy of a good sold every day via a common posted price for all 
buyers whose values are all drawn iid from a known distribution $F$. If more 
than one buyer is interested in purchasing at the posted price, 
then the item is awarded to a uniformly random buyer. In this setting, the 
paper shows that the seller can achieve a revenue that is a constant fraction 
of the full commitment benchmark in a threshold strategy PBE. The fact that the 
seller is forced to post a common price for multiple buyers means that the 
ability to use the reveled preference of the buyers is diminished. This is one 
of the key differences from the single buyer setting we study. A different but related line of work on repeated auctions is one where the buyer's value is drawn independently from a distribution in every round (unlike our identical value in all rounds but is drawn once initially from a distribution); see~\cite{AminRs2013, PPPR16, ADH16, MLTZ16, MLTZ18, ADMS18, BMSW18} for more on this line of work. Another line of work on repeated sales with different values across rounds is one where the buyer's value for the good evolves with time/usage: see~\citep{KLN13, CDKS16}.

\paragraph{Organization.} In Section~\ref{sec:prelim}, we formally define the
game, Perfect Bayesian Equilibrium (PBE) and related concepts. In
Section~\ref{sec:zeroCommitment}, we study the finite horizon game with no
commitment. In Section~\ref{sec:coupons}, we study the finite horizon partial
commitment game. In Section~\ref{sec:infiniteHorizon}, we consider the time
discounted infinite horizon game with partial commitment and its connections to 
bargaining and durable goods monopoly. In Section~\ref{sec:conc} we conclude 
with some open problems.

\section{Preliminaries}
\label{sec:prelim}
\paragraph{Bayesian Nash Equilibrium.} The most common notion of equilibrium in a
static game of incomplete information is the Bayesian Nash Equilibrium. A
profile of strategies is a Bayesian Nash Equilibrium (BNE) if for every agent,
given the other agents' strategies, his own strategy maximizes his expected
payoff for each of his type. The expected payoff of an agent is computed using
the agent's beliefs about the private types of other agents, and all the agents'
beliefs are assumed to be consistent with a common prior distribution over all
the private types. 

\paragraph{History.}
The game proceeds over $\rounds$ rounds, and each round consists of two stages:
round $r$ consists of stages $\stage=2r-1$ and $\stage=2r$. At $\stage=2r-1$, the seller sets a
price, and at $\stage=2r$, the buyer reacts with a accept or reject. The history
after $\stage$ stages of play is denoted by $\hist^\stage$, and constitutes the prices
and accept/reject decisions of all stages $k': 0 \leq k' \leq \stage$.

\paragraph{Beliefs and buyer types.}
In our game, since the buyer's type alone is private, the seller alone has a
belief over the buyer's private type. The seller's belief
$\bel(\cdot|\hist^\stage)$
is a probability density function over the buyer's private type. At the 
beginning of the game, the buyer types are assumed to be drawn from a publicly 
known, atomless 
bounded support distribution in $[\low,\high]$.

\paragraph{Strategy spaces.}
The seller's action space is restricted to posting a
non-negative price in every round. Correspondingly, the seller's strategy
$\mstrats(\cdot|\hist^\stage)$ is a function that, for every possible history,
outputs a probability distribution over his available actions (i.e.,
non-negative prices). The buyer's action space is
restricted to accepting or rejecting a price. Correspondingly, the buyer's
strategy $\mstratb(\cdot|\val,\hist^\stage)$ is a function that, for every
possible private value of the buyer and every possible history, outputs a
probability for accepting the item at the posted price. 

\paragraph{Perfect Bayesian Equilibrium.}  Intuitively a Perfect Bayesian
Equilibrium combines the notions of subgame perfect equilibrium (used in dynamic
games of complete information) and Bayesian update of beliefs (used in games of
incomplete information) by requiring that the profile of strategies and beliefs
when applied to the continuation game given any history, form a BNE. It is the
perfection aspect of PBE that makes commitments non-credible/non-binding:
informally, no commitment is credible unless it is a part of a BNE in the
continuation game after every possible history that could precede the stage at
which the commitment becomes effective. We now formally define Perfect Bayesian
Equilibrium for our game, i.e., mention only the restrictions relevant to our
game. 

A profile of strategies
$(\mstrats^*(\cdot|\hist^\stage),\{\mstrat_v^*(\cdot|\val,\hist^\stage)\}_{v})$ 
and
beliefs $\bel(\cdot|\hist^\stage)$ in the repeated-sale game is a Perfect
Bayesian Equilibrium (PBE) when the following conditions are satisfied:
\begin{enumerate}
	\item Bayesian update of seller's beliefs: the seller assumes that the buyer
	plays the PBE strategy $\mstrat_v^*(\cdot|\val, \hist^\stage)$. If there exists 
	some value
	$\val$ in the support of seller's belief $\bel(\cdot|\hist^{\stage-1})$, such 
	that the buyer's action at
	stage $\stage$ has a non-zero probability under his equilibrium strategy
	$\mstrat_v^*(\cdot|\val,\hist^\stage)$ at $\val$, the seller updates his
	belief $\bel(\cdot|\hist^\stage)$ based on Bayes' rule.  There are no
	restrictions on belief updates if the buyer takes an out-of-equilibrium 
	zero-probability action.
	
	\item For every $\stage$ and $\hist^\stage$, the strategies from $\hist^\stage$
	onwards are a BNE for the remaining game. Formally, conditional on reaching
	$\hist^\stage$, let 
	$u_s\left(\mstrats(\cdot|\hist^\stage),\{\mstrat_v(\cdot|v,\hist^\stage)\}_{v},\bel(\cdot|\hist^\stage)\right)$
	denote
	the expected revenue of seller under strategy profile $\mstrat$ (where the
	expectation is over both the randomness in $\mstrat$ and the belief
	$\bel(\cdot|\hist^\stage)$), and let 
	$u_v\left(
	\mstrats(\cdot|\hist^\stage),\{\mstrat_v(\cdot|v,\hist^\stage)\}_{v},\bel(\cdot|\hist^\stage)\right)$
	denote the
	expected utility of the buyer type $v$ under strategy profile $\mstrat$ (where 
	the
	expectation is over the randomness in $\mstrat$). Then,
	\begin{align*}
	& 
	u_s\left(\mstrats^*(\cdot|\hist^\stage),\{\mstrat_v^*(\cdot|v,\hist^\stage)\}_{v},\bel(\cdot|\hist^\stage)\right)\geq
	u_s\left(\mstrats(\cdot|\hist^\stage),\{\mstrat_v^*(\cdot|v,\hist^\stage)\}_{v},\bel(\cdot|\hist^\stage)\right)\qquad
	\forall k, \forall \hist^\stage, \forall \mstrats\\
	&u_v\left(\mstrats^*(\cdot|\hist^\stage),\{\mstrat_v^*(\cdot|v,\hist^\stage)\}_{v},\bel(\cdot|\hist^\stage)\right)\geq
	u_v\left(\mstrats^*(\cdot|\hist^\stage),\mstrat_v(\cdot|v,\hist^\stage),\{\mstrat_x^*(\cdot|x,\hist^\stage)\}_{x\neq
		v},\bel(\cdot|\hist^\stage)\right)\qquad
	\forall v, \forall k, \forall \hist^\stage, \forall \mstrat_v
	\end{align*}
\end{enumerate}

\paragraph{Threshold PBE.} A threshold strategy for the buyer computes an
accept/reject decision as follows: given history $\hist^\stage$, there exists a 
deterministic threshold $t(\hist^\stage) \geq 0$, and buyer type $v$ accepts 
the item if $\val \geq t(\hist^\stage)$ and rejects otherwise.  By definition, a
threshold strategy is a pure strategy. In this paper, we focus on pure-strategy 
threshold 
strategy PBEs, i.e., threshold strategy for the buyer and pure strategy for the 
seller. 

\paragraph{Single buyer vs. continuum of buyers.} As discussed in many prior 
works (e.g. see Section 10.2.1 of~\cite{FT91}), given a distribution of 
types, one can interpret our 
problem as a single seller against a single buyer with a private type drawn 
from this type space, or a continuum of buyer types whose values are given by 
the type distribution. In the latter case, we assume that the seller cannot 
distinguish between buyer types and can simply observe the measures of buyer 
sets that accept or reject. For most of this paper we use the single buyer 
interpretation. 

\paragraph{Simultaneous deviation of several buyer types.} Following the 
convention from previous works (\citet{Gul1986}), we do not specify the 
equilibrium 
behavior following simultaneous deviation by several buyers. In the single 
buyer interpretation, this does not matter because simultaneous deviations are 
not visible to the seller, but even in the continuum of buyers model where it 
could be 
observable if a non-zero measure of buyers deviate we do not specify 
equilibrium behavior following such deviations.

\paragraph{The no-gap case.} We assume that the seller has a fixed marginal 
cost for producing each copy of the good, and (without loss of generality) that 
it is $0$. Prior work has distinguished two important cases of buyer 
distributions as a function of the seller cost: whether there is a gap (the gap 
case) between $\ell$ and $0$ or whether $\ell = 0$ (the no-gap case). We focus 
on the (more natural) no-gap case i.e., $\ell = 0$ in 
Section~\ref{sec:zeroCommitment}. The threshold PBE existence result in 
Section~\ref{sec:coupons} applies for both the gap and no-gap case.  

\paragraph{PBE specification.} As is clear from the definition of a PBE, 
seller's and buyer's strategies should be specified in off-equilibrium 
histories as well, for a strategy profile to constitute a PBE. In the 
theorems we prove, the off-equilibrium behavior can be immediately derived from 
applying the fact that the threshold buyer is indifferent between buying and 
rejecting. Thus we skip the excessively long specification of strategies after 
all off-equilibrium histories. For the $2$ round case alone 
(Section~\ref{subsec:twoRounds}), we specify strategies in complete detail for 
the 
sake 
of clarity. 

\paragraph{Notational convention for thresholds.} We use $t=\infty$ to denote 
the
buyer rejecting to buy at all values. Thus, a threshold $t$ lies in
$[\low,\high] \cup \{\infty\}$. 

\section{Finite horizon with no commitment}
\label{sec:zeroCommitment}
We begin with two simple facts that are useful in analyzing PBEs. Their proofs 
are immediate. 
\begin{fact} (Indifference at threshold)
	\label{lem:IAT}
	In any threshold PBE, the buyer with his value $\val$ equals the threshold $t$  
	is indifferent between accepting and rejecting (except for $t = \infty$). 
\end{fact}

\begin{fact} (Bayesian price update)
	\label{lem:BPU}
	If the buyer accepts in a given round with threshold $t$ all future round 
	prices are at least $t$, and if he rejects in a given
	round with threshold $t$, all future round prices
	are at most $t$.
\end{fact}

\subsection{Two rounds game}
\label{subsec:twoRounds} 
It turns out that for a two rounds game, a
threshold PBE is guaranteed to exist and it is essentially
unique.~\citet{HT88} and~\citet{FV06} characterize the PBE for the two rounds
repeated sales game. For the sake of completeness and for gaining intuition,
and because our main result uses the $2$ rounds PBE (mildly), we
discuss the $2$ rounds PBE now. 

\paragraph{Notation.} 
We begin with some
notation and two quick definitions.  Let $\distab$ denote the distribution on
$\val$ conditioned on the fact that $a \leq \val \leq b$ (and thus $\dist =
\distlh$). Let $\psab$ denote an arbitrary element of $\argmax_p
p(1-\distab(p))$ i.e., the set of all single-round revenue maximizing prices or
the so called {\em monopoly prices} for $\distab$. Let $\ps = \pslh$.  Whenever
the monopoly price is not unique $\psab$ will denote an arbitrary monopoly
price unless specified otherwise.  

\paragraph{Revenue Curve.} The revenue curve $\Rab(p) = p(1-\distab(p))$ at $p$
gives the expected revenue in a single round game obtained by offering a price
$p$ to a buyer whose value is drawn from $\distab$. Let $R(\cdot) = \Rlh(\cdot)$
denote the revenue curve for the distribution $\dist$. 

We prove here a property that any PBE in two rounds game for any atomless
bounded support distribution must satisfy.

In the following lemma, $p_1$ is the price in the first round, $\prej$ and 
$\pacc$ are the
prices in second round, upon buyer's rejection and acceptance respectively in
the first round, given that the first round price is $p_1$.  We will use 
$t(p_1)$ to denote the threshold used by the buyer in first round. An 
application of Fact~\ref{lem:BPU} shows that the price $p_{21}$ in the second 
round is at least $t(p_1)$. A further application of Fact~\ref{lem:IAT} shows 
that $p_{20} = p_1$ because indifference implies that $t(p_1) - p_{20} = 
t(p_1)- p_1$. 

\begin{lemma}\label{thm:twoRounds}
	For any atomless distribution $\dist$ of 
	buyer's value supported in $[\ell = 0,\high]$, every pure strategy threshold 
	PBE of a two rounds repeated sales game
	will have $p_1$ such that $\ell < t(p_1) < \high$.
\end{lemma}
\begin{proof}
	First note that $t(p_1)$ is unique. If not, there will be two thresholds 
	$t(p_1) < t'(p_1)$
	such that $\psab[\low,t(p_1)]= \psab[\low,t'(p_1)] = p_1$. I.e., the virtual 
	values for $\distlt$ and $\distab[\low,t']$, namely $\phi_{\distlt}$ and
	$\phi_{\distab[\low,t']}$ {\em both} become zero at $p_1$.  This is not possible
	because for any $x \leq t$, we have $\distab[\low,t'](x) = 
	\distlt(x)\cdot\alpha$ for
	some $\alpha < 1$. This means, $\phi_{\distab[\low,t']}(x) < \phi_{\distlt}(x)$
	for all $x \leq t$. Therefore both the virtual value functions cannot become
	zero at the same point. 
	
	We will prove the lemma by showing that when $t(p_1) = \low$ or $\high$, the 
	seller's revenue is exactly $R(p^*)$, i.e., the monopolist's single round 
	revenue from $\dist$. We show that the seller can do strictly better than 
	$R(p^*)$ in a two 
	round PBE with $\low < t(p_1) < \high$. The former statement is immediate: when 
	all buyers reject in the first round, the first round revenue is $0$, and the 
	maximum possible revenue from just second round alone is $R(p^*)$. Similarly, 
	when all buyers accept in the first round, the first round price $p_1$ must 
	have been $0$ (otherwise buyers with $v < p_1$ would have incurred 
	negative utility), and here again, all revenue comes from second round, which 
	is at most $R(p^*)$. 
	
	Whenever $\low < t(p_1) < \high$, note that the seller's revenue is exactly 
	$R(\prej) +R(\pacc)$: the buyer buys once when his value 
	exceeds $\prej$ and once more when his value exceeds $\pacc$. Use this to see 
	that if $\low < t(p_1) \leq 
	p^*$, seller's revenue is strictly larger than $R(p^*)$: because $\pacc 
	= p^*_{[t(p_1),\high]} = p^*_{[\low,\high]} = p^*$ (the last but one equality 
	holds for all truncated distributions of the form $\dist_{[x,\high]}$ where $x 
	\leq p^*$: this 
	is true in this case since $t(p_1) \leq p^*$), 
	and 
	$\prej = 
	p^*_{[\low,t(p_1)]}$, which yields a revenue of $R(p^*_{[\low,t(p_1)]}) + 
	R(p^*) > R(p^*)$. This completes the proof. 
	
	There is one detail to fill: is it always 
	possible to find a $p_1$ such that $\low < t(p_1) \leq p^*$?  Since $\ell = 0$, 
	it 
	follows that $p^* > \ell = 0$ and clearly setting $t(p_1) = x$ for any $x$ s.t. 
	$0 < x < p^*$ will do, since the distribution $\dist_{[0,x]}$ will have a 
	non-zero monopoly 
	price, and that will be $p_1$. 
\end{proof}

\paragraph{Solution to the general 2 rounds repeated sales game.} The proof of 
Lemma~\ref{thm:twoRounds} gives us a recipe for constructing a 2 rounds PBE for 
an arbitrary atomless bounded support distribution. As discussed in 
Lemma~\ref{thm:twoRounds}'s proof, the seller's goal is to maximize $R(\prej) 
+R(\pacc)$. Because of Fact~\ref{lem:IAT}, it follows that $p_1 = p_{20}$ in 
any PBE. Thus the seller has to compute a first round price of $z$ such that 
$R(\prej) 
+R(\pacc) = R(z) + R(p^*_{[t(z),1]})$ is maximized. This is basically what is 
captured in the strategies in Algorithms ~\ref{strat:SellerTwoRoundsArbitrary} 
and~\ref{strat:BuyerTwoRoundsArbitrary}. Note that $p_1 > p^*$ cannot be in 
the 
equilibrium path: if $p_1 > p^*$, the equilibrium behavior will be for all 
buyers to reject; if on the contrary we had $t(p_1) < \high$, the price 
$p_{20}$ in 
the second round is the monopoly price of the 
distribution $\dist_{[\ell,t(p_1)]}$ which is strictly lesser than $p^*$, 
thereby making the threshold buyer non-best-responding in this case).

For concreteness, we discuss the special case of $U[0,1]$ in 
Appendix~\ref{app:twoRounds}.

\IncMargin{1em}
\begin{algorithm}[!h]
	\SetKwData{Left}{left}\SetKwData{This}{this}\SetKwData{Up}{up}
	\SetKwFunction{Union}{Union}\SetKwFunction{FindCompress}{FindCompress}
	\SetKwInOut{Input}{Input}\SetKwInOut{Output}{Output}
	\textbf{\underline{Round-$1$ strategy:}}\\
	Let $t(x)$ be such that $p^*_{[\low,t(x)]} = x$\\
	Set $p_1 = \argmax_{z \leq p^*} \{R(z) + R(p^*_{[t(z),1]})\}$\\
	\textbf{\underline{Round-$2$ strategy:}} \\
	\uIf{$p_1 \leq p^*$}
	{
		\uIf{\text{Buyer rejects in round $1$}}
		{
			Set second round price of $p_2 = p_1$
		}
		\ElseIf{\text{Buyer accepts in round $1$}}
		{
			Set second round price of $p^*_{[t(p_1),1]}$
		}
	}
	\ElseIf{$p_1 > p^*$}
	{
		\uIf{\text{Buyer rejects in round $1$}}
		{
			Set second round price of $p_2 = p^*$
		}
		\ElseIf{\text{Buyer accepts in round $1$}}
		{
			Set second round price of $\high$
		}
	}
	\caption{Seller's strategy in the general $2$ rounds game}
	\label{strat:SellerTwoRoundsArbitrary}
\end{algorithm}
\DecMargin{1em}
\IncMargin{1em}
\begin{algorithm}[!h]
	\SetKwData{Left}{left}\SetKwData{This}{this}\SetKwData{Up}{up}
	\SetKwFunction{Union}{Union}\SetKwFunction{FindCompress}{FindCompress}
	\SetKwInOut{Input}{Input}\SetKwInOut{Output}{Output}
	\textbf{\underline{Round-$1$ strategy:}}\\
	\uIf{$p_1 \leq p^*$}
	{
		Let $t(p_1)$ be such that $p^*_{[\low,t(p_1)]} = p_1$\\
		\uIf{$v \geq t(p_1)$}
		{
			Accept
		}
		\Else
		{
			Reject
		}
	}
	\ElseIf{$p_1 > p^*$}
	{
		Reject
	}
	\textbf{\underline{Round-$2$ strategy:}} \\
	\uIf{$v \geq p_2$}
	{
		Accept
	}
	\Else
	{
		Reject
	}
	\caption{Buyer's strategy in the general $2$ rounds game}
	\label{strat:BuyerTwoRoundsArbitrary}
\end{algorithm}
\DecMargin{1em}

\subsection{General $n$ rounds game.} We now move to the main result of this
section: in a $n$ rounds repeated sales game for $n > 2$, a pure strategy 
threshold PBE never exists. 
\begin{theorem}\label{thm:nRounds}
	For any atomless distribution $\dist$ of 
	buyer's value supported in $[\ell = 0,\high]$ and for any $n > 2$, a pure 
	strategy threshold PBE never exists in a $n$ rounds repeated 
	sales game. 
\end{theorem}

\begin{proof}
	Consider the three rounds case first. Let $p_1$ be the first round price. Let
	$t = t(p_1)$ be the corresponding buyer threshold in the first round.  Note 
	that a PBE requries 
	that given any history, the strategies for the continuation game must be 
	mutually best responding. Consider one such history where the first round price 
	is $p_1 > \low$ (we are not fixing on the equilibrium $p_1$, but an arbitrary 
	first round price $p_1$ --- recall that for every first round $p_1$ a PBE must 
	specify 
	equilibrium behavior in the continuation game). We show in three cases that 
	irrespective of what value $t(p_1)$ takes, we cannot have a threshold PBE with 
	$p_1 > \ell$. 
	
	\textbf{Case 1: When $p_1 > \ell$, $t(p_1) = \ell$ is not possible in PBE.} 
	Clearly 
	this is 
	not possible as the buyer with value $\ell$ gets negative utility in the first 
	round, and $0$ utility in the future rounds since price is guaranteed to be at 
	least $\ell$ in future rounds. 
	
	\textbf{Case 2: When $p_1 > \ell$, $\ell < t(p_1) < \high$ is not possible in 
		PBE.} 
	We show that the threshold buyer 
	is never indifferent between accepting and rejecting in first round, violating 
	Fact~\ref{lem:IAT}. Begin by noting that $t(p_1) \geq p_1$. 
	\begin{enumerate}
		\item If $t(p_1) = p_1$, the 
		threshold buyer makes $0$ utility upon acceptance, where as by rejecting he can 
		make a non-negative 
		utility as he will be the largest point in the support of the seller beliefs 
		after 
		rejection, namely $F_{[\low,t]}$. 
		\item The only remaining case is that $p_1 < t(p_1) 
		< \high$. In this case, the threshold buyer upon accepting the price $p_1$ in 
		the first 
		round cannot get any further utility in the future rounds, as the prices are at 
		least $t(p_1)$ in the future given seller beliefs are $[t(p_1),\high]$ after 
		first round. Thus his total utility is $t(p_1)-p_1$ upon 
		accepting in the first round.  We show that by rejecting in the first round, 
		the threshold buyer obtains strictly larger utility which is a 
		contradiction. Let $\prej$ be the price in the second round on rejection, and 
		let $\prejrej$ and $\prejacc$ denote the price in the third round upon $(reject,
		reject)$ and $(reject, accept)$ respectively in the first two rounds.  When a
		buyer with value $t$ rejects in the first round, and accepts in the second and
		third rounds, he gets a utility of $(t(p_1)-\prej) + (t(p_1)-\prejacc)$. The 
		two claims
		below show that $\prej \leq p_1$ and $\prejacc < t(p_1)$. Therefore the sum
		$(t(p_1)-\prej) + (t(p_1)-\prejacc)$ is strictly larger than $t(p_1)-p_1$. 
	\end{enumerate}
	
	\paragraph{Claim 1: $\prej \leq p_1$.} On the contrary suppose that $\prej >
	p_1$. Consider a buyer value with $\val$ s.t.
	$p_1 < \val < \prej$. Such a buyer gets zero utility upon rejection in the
	first round because all prices  after rejection are strictly larger than his
	value $\val$ (because the second round price $\prej > \val$ by our choice of
	$\val$, and a threshold PBE for the remaining two rounds game implies that 
	$\prej 
	= \prejrej \leq \prejacc$. The equality follows from applying 
	Fact~\ref{lem:IAT} to the threshold buyer in second round, and inequality 
	follows from noting that the third round price of seller is at least second 
	round threshold which is at least second round price of $p_{20}$). Where as
	upon acceptance, he would have a gotten a strictly positive utility of $\val -
	p_1$.  This says that $\val$ accepts in first round, i.e., $v \geq t(p_1)$. But 
	this is a contradiction because $\val < \prej$ by our choice of $\val$, and,  
	$\prej
	\leq t(p_1)$ because the belief after first round is $F_{[\low,t(p_1)]}$. 
	
	\paragraph{Claim 2: $\prejacc < t(p_1)$.} We show that all prices after
	rejection in first round, namely, $\prej,
	\prejrej, \prejacc$ are strictly smaller than $t(p_1)$.  Even if the largest 
	among 
	these prices, namely
	$\prejacc$, was equal to $t(p_1)$, that would not be a PBE. To see this, 
	consider
	the threshold $t'$ used by the buyer for the distribution 
	$\dist_{[\ell,t(p_1)]}$ in the two round continuation game following the first 
	round price of $p_1$. By
	Lemma~\ref{thm:twoRounds} such a threshold $t'$ is strictly smaller than the 
	largest point in the support $t(p_1)$.
	Note that $\prejacc$ is simply be the monopoly price for the
	distribution $\distab[t',t(p_1)]$. If $t' < t(p_1)$, this monopoly price 
	$\prejacc$ cannot 
	be $t(p_1)$ because that yields a $0$ revenue which is not optimal. Thus 
	$\prejacc < t(p_1)$. 
	
	\textbf{Case 3: When $p_1 > \ell$, $t(p_1) = \high$ is not possible in PBE.} For
	a price $p_1 > \ell$ to be rejected by all buyer types in the first round (i.e.,
	for having $t(p_1) = \high$), we need $\prej \leq p_1$ (for otherwise a buyer
	with value $\low < p_1 < \val < \prej$ would not be best responding by rejecting
	in the first round). I.e., for a PBE to exist, for every $p_1 > \ell$, there
	should be a two round threshold PBE for $\dist$ with a first round price $\prej
	\leq p_1$ for every $p_1 > \ell$. But when $p_{20}$ grows arbitrarily close to
	$0$, the revenue in the last two rounds game becomes arbitrarily close to the
	last round revenue --- on the other hand, from the proof of
	Lemma~\ref{thm:twoRounds}, we know that the seller in a two rounds PBE can
	always get a revenue strictly better than and bounded away from the single round
	revenue, and hence having $p_{20}$ arbitrarily close to $0$ cannot be a part of
	a PBE.
	
	\paragraph{Three rounds to $n$ rounds.} If a three rounds threshold PBE cannot 
	exist, neither can an $n$-rounds threshold PBE.  
\end{proof}

\section{Finite horizon with partial commitment}
\label{sec:coupons}
While a threshold PBE never exists when there is no commitment from the 
seller's side, things
change dramatically if we allow partial commitment.  We show that by having the 
partial commitment from the seller of not raising prices upon purchase, a pure 
strategy threshold PBE is guaranteed to exist for all distributions (both the 
gap-case and no-gap case, namely $\ell$ is not necessarily $0$). 

\paragraph{Numbering convention.} For this section alone we change our 
convention for numbering rounds compared to what we used in previous sections: 
the price, and threshold in the first round of the $\rounds$
rounds game are denoted by $p_\rounds$ and $t_\rounds$ (in earlier sections we
used $p_1$ and $t_1$ for the first round). Similarly the second round's
corresponding quantities are $p_{n-1}$ and $t_{n-1}$ and so on.  

\begin{theorem}
	\label{thm:partialCommitmentPBEExistence}
	For any atomless distribution $\dist$ of 
	buyer's value supported in $[\ell, \high]$, and for any $n$, a pure 
	strategy threshold PBE always exists in a $n$ rounds partial commitment 
	repeated 
	sales game. 
\end{theorem}
\begin{proof}
	We prove by induction on the number of rounds $r$ that a pure strategy 
	threshold PBE is guaranteed 
	to exist for $\dist_{[\ell,x]}$ for all $\ell \leq x \leq \high$. 
	
	\underline{Base case:} When $r=1$, a pure strategy threshold 
	PBE trivially exists for all $\dist_{[\ell,x]}$: the seller posts the monopoly 
	price of the distribution 
	$\dist_{[\ell,x]}$, and buyers with values at least the monopoly price of 
	$\dist_{[\ell,x]}$ accept that price. 
	
	\underline{Inductive hypothesis:} Assume that for all $r \leq n-1$ a threshold 
	PBE exists for $\dist_{[\ell,x]}$ for all $\ell \leq x \leq \high$ in a $r$ 
	rounds partial commitment game. 
	
	\underline{Inductive step:} Consider $r=n$. Let $p_{r,[a,b]}$ denote the 
	first 
	round PBE price in a $r$-rounds game for distribution $\dist_{[a,b]}$ (if 
	more than one PBE exists, fix an arbitrary one). We use $p_r$ to denote 
	$p_{r,[\low,\high]}$.
	
	\paragraph{Price remains fixed upon acceptance.} We begin by showing that when 
	a buyer 
	accepts a price of $p_n$ in the first round, the price remains $p_n$ for the 
	remaining $n-1$ rounds. Clearly by the definition of the game, the price cannot 
	increase beyond $p_n$. Also, the price will not decrease below $p_n$ because, 
	if $t_n$ is the threshold used by the buyer in the first round ($t_n 
	\geq p_n$), the future belief of the seller will be $[t_n,\high]$, i.e., the 
	smallest 
	point $t_n$ in the support of the distribution  is larger than $p_n$, 
	and therefore $p_n$ will be accepted in all future rounds. 
	
	\paragraph{Continuation game has a PBE after first round price of $p_n$.} We 
	now show that for each value of the first round price $p_n$, the continuation 
	game has a PBE. Let $t_n(p_n)$ (abbreviated to just $t_n$) be the buyer's 
	threshold while facing 
	a price of $p_n$. Indifference at threshold (Fact~\ref{lem:IAT}) implies that
	\begin{align}
	\label{eqn:nthRoundIndifference}
	n(t_n - p_n) &= k\cdot(t_n - p_{k,[\ell,t_n]}) 
	\end{align}
	Here the LHS in~\eqref{eqn:nthRoundIndifference} is the utility of the 
	threshold 
	buyer upon accepting a price of $p_n$ in the first round, and the RHS is his 
	utility upon 
	rejecting $p_n$ in the first round. The LHS utility is clear: once a price of 
	$p_n$ is accepted, it remains so for all the remaining rounds, giving a utility 
	of $n(t_n - p_n)$. For the RHS utility, note that after a threshold of $t_n$ in 
	the first round, the seller's belief is updated to $\dist_{[\ell,t_n]}$ after 
	rejection. The PBE of $\dist_{[\ell,t_n]}$ could possibly have several 
	``wasted'' rounds where all buyers reject, and finally at some round there is a 
	non-trivial threshold: this is why we have a $k$ in the RHS instead of $n-1$, 
	and $k \in \{1,2,\dots,n-1\}$. 
	
	As a sanity check, note that we require $t_n \geq p_n$, and this is indeed 
	immediate from~\eqref{eqn:nthRoundIndifference} because $p_{k,[\ell,t_n]} \leq 
	t_n$. The seller now maximizes his revenue $R_n$ from n rounds as
	
	\begin{align}
	\label{eqn:nRoundRevMax}
	R_n = \max_{p_n} \bigg\{R_{n-1,[\ell,t_n]}\dist(t_n) + np_n(1 - 
	\dist(t_n))\bigg\}
	\end{align}
	
	Given that $t_n$ can be obtained as a function of $p_n$ 
	using~\eqref{eqn:nthRoundIndifference} and that $R_{n-1,[\ell,t_n]}$ is well 
	defined by induction to be the PBE revenue in $n-1$ rounds for the distribution 
	$\dist_{[\ell,t_n]}$, it follows that the seller has a well-defined 
	maximization problem in~\eqref{eqn:nRoundRevMax}, and picks the price $p_n$ to 
	maximize his revenue. 
	
	This proves the existence of a PBE in $n$ rounds for the 
	distribution $\dist$. An identical argument will show the existence of $n$ 
	rounds PBE for $\dist_{[\ell,x]}$ for all $\ell \leq x \leq \high$. 
\end{proof}

\paragraph{The $U[0,1]$ case and uniqueness.} It is not possible to prove 
uniqueness at the level of all atomless distributions. For the sake of 
concreteness, and to illustrate the difference partial commitment can make, we 
focus on the $U[0,1]$ distribution and compute a threshold PBE that obtains a 
revenue of $\sqrt{\frac{n}{2}+\frac{\log n}{8}+O(1)}$ in a $n$ rounds game, and 
also establish its 
uniqueness. In comparison, the no commitment $n$ rounds game does not even have 
a threshold PBE. 

\begin{theorem}\label{thm:recursion}
	For the $U[0,1]$ distribution, the $n$ rounds partial commitment repeated sales 
	game has a unique\footnote{
		the uniqueness here refers to the equilibrium path and ignores irrelevant 
		multiplicities that arise in off-equilibrium path. Note that there are trivial 
		ways to have multiplicity in the off equilibrium 
		path. For instance, when the seller posts a price of $0$ in a round (an 
		off-equilibrium path), the only 
		possible PBE behavior in the continuation game is for all buyers to accept the 
		price. If some buyer doesn't, that constitutes a zero probability action and 
		the seller is free to update beliefs arbitrarily. At this node, any probability 
		distribution over $[0,1]$ as belief can still support the equilibrium path 
		behavior discussed in the theorem. We ignore these trivial and irrelevant 
		multiplicities in 
		off-equilibrium path.}
	pure strategy threshold PBE that obtains a revenue of 
	$\sqrt{\frac{n}{2}+\frac{\log n}{8}+O(1)}$. 
\end{theorem}

\begin{proof}
	We prove by induction on the number of rounds $r$ that for each $x$ in 
	$(0,1]$, the distribution $\dist_{[0,x]}$ satisfies the following:
	\begin{enumerate}
		\item the $r$ rounds partial commitment game on $\dist_{[0,x]}$ has a pure 
		strategy threshold PBE ;
		\item the threshold PBE is unique;
		\item the PBE threshold $t_r$ of the buyer in the first round of the $r$ rounds 
		game is non-trivial (i.e., not 
		at the end-points of the support), namely, $0 < t_r < x$.  
	\end{enumerate}
	
	\underline{Base case:} When $r=1$, a pure strategy threshold 
	PBE trivially exists for all $\dist_{[0,x]}$: the seller posts the monopoly 
	price $x/2$, and buyers with values at least $x/2$ accept that price. 
	Uniqueness is obvious from the uniqueness of $\argmax_{p}p(x-p)$. The threshold 
	of $x/2$ is non-trivial, i.e., $0 < x/2 < x$.
	
	\underline{Inductive hypothesis:} Assume that for all $r \leq n-1$ a threshold 
	PBE exists for $\dist_{[0,x]}$ for all $0 < x \leq 1$ in a $r$ 
	rounds partial commitment game, and that it is unique, with a non-trivial 
	threshold $t_r$ in the first round. 
	
	\underline{Inductive step:} Consider $r=n$. Let $p_{r,[a,b]}$ denote the 
	first 
	round PBE price in a $r$-rounds game for distribution $\dist_{[a,b]}$. We use 
	$p_r$ to denote 
	$p_{r,[0,1]}$.
	
	\paragraph{Price remains fixed upon acceptance.} Identical to our proof in 
	Theorem~\ref{thm:partialCommitmentPBEExistence}, it holds that when 
	a buyer accepts a price of $p_n$ in the first round, the price remains $p_n$ 
	for the remaining $n-1$ rounds. 
	
	\paragraph{Continuation game has a PBE after first round price of $p_n$.} Just 
	like in Theorem~\ref{thm:partialCommitmentPBEExistence}, applying indifference 
	at threshold (Fact~\ref{lem:IAT}) implies that
	\begin{align}
	\label{eqn:nthRoundIndifferenceUniform}
	n(t_n - p_n) &= (n-1)\cdot(t_n - p_{n-1,[0,t_n]}) 
	\end{align}
	The only difference from~\eqref{eqn:nthRoundIndifference} is that instead of an 
	arbitrary number $k \leq n-1$, the RHS now has exactly $n-1$: this is because 
	by inductive hypothesis, we have that the threshold is non-trivial in the first 
	round of an $n-1$ round game. So for the $U[0,t_n]$ distribution that is left 
	after rejecting in the first round, a non-trivial 
	threshold implies that at least the largest point in the support of the 
	distribution, namely $t_n$, buys in the second round, in which case his utility 
	is $(n-1)\cdot(t_n - p_{n-1,[0,t_n]})$. 
	
	To simplify this further, note that $U[0,t_n]$ is simply a scaled version of 
	$U[0,1]$. Thus, $p_{n-1,[0,t_n]} = t_n\cdot p_{n-1,[0,1]} = t_n\cdot p_{n-1}$. 
	Thus, we can rewrite~\eqref{eqn:nthRoundIndifferenceUniform} as
	
	\begin{align}
	\label{eqn:nthRoundIndifferenceUniformRewrite}
	n(t_n - p_n) &= (n-1)t_n\cdot(1 - p_{n-1}) 
	\end{align}
	
	Let $u_n$ be the PBE utility of the agent with value $1$ in a $n$ rounds game. 
	Clearly,
	\begin{align}\label{eqn:un}
	u_n = n(1-p_n).
	\end{align}
	By this definition of $u_n$, 
	equation~\eqref{eqn:nthRoundIndifferenceUniformRewrite} can be rewritten as 
	\begin{align}\label{eqn:tn}
	n(t_n - p_n) = u_{n-1}t_n. 
	\end{align}
	
	The seller's revenue, much like~\eqref{eqn:nRoundRevMax} can be written as
	
	\begin{align}
	\label{eqn:rn}
	R_n &= \max_{p_n}\bigg\{R_{n-1, [0,t_n]}t_n + (1-t_n)\cdot np_n\bigg\} 
	\nonumber\\ 
	&= \max_{p_n}\bigg\{R_{n-1}t_n^2 + (1-t_n)\cdot np_n\bigg\}. 
	\end{align}
	
	We now have a four variable recurrence in 
	$\{u_n, t_n, p_n, R_n\}$ to solve, given by 
	equations~\eqref{eqn:un},~\eqref{eqn:tn},~\eqref{eqn:rn}. Substituting for 
	$p_n$ from equation~\eqref{eqn:tn} into
	equation~\eqref{eqn:rn} we have
	\begin{align}
	\label{eqn:rnOptimize}
	R_n = \max_{t_n}\bigg\{R_{n-1}t_n^2 + (1-t_n)t_n(n-u_{n-1})\bigg\} . 
	\end{align}
	This is an expression for revenue that the seller has to maximize. Notice that
	$R_{n-1}$ and $u_{n-1}$ are fixed quantities that are not to be optimized:
	these are quantities for the $n-1$ rounds game for which we assume by induction
	that there is a unique threshold PBE and hence revenue, utilities etc. are
	fixed. The only quantity to optimize in this expression is $t_n$. This
	expression for $R_n$ is maximized at $t_n =
	\frac{n-u_{n-1}}{2(n-u_{n-1}-R_{n-1})}$. Substituting this value of $t_n$ into equation~\eqref{eqn:tn},
	we get 
	\begin{align}
	\label{eqn:pnFinal}
	p_n &= \frac{(n-u_{n-1})^2}{2n(n-u_{n-1}-R_{n-1})} .
	\end{align}
	Similarly, substituting $t_n$ into equation~\eqref{eqn:rnOptimize}, we get
	\begin{align}
	\label{eqn:rnFinal}
	R_n &= \frac{(n-u_{n-1})^2}{4(n-u_{n-1}-R_{n-1})}.
	\end{align}
	From equations~\eqref{eqn:pnFinal} and~\eqref{eqn:rnFinal} it is easy to verify 
	that 
	\begin{align}
	\label{eqn:Rnpn}
	R_n = \frac{np_n}{2}.
	\end{align}
	Using~\eqref{eqn:Rnpn}, and combining equations~\eqref{eqn:un} 
	and~\eqref{eqn:rnFinal}, we eliminate three out of four variables
	to get
	\begin{align*}
	R_n = \frac{(1+2R_{n-1})^2}{4(1+R_{n-1})}.
	\end{align*}
	To analyze this recursion, substitute $V_n = R_n+1$. This yields 
	\begin{align*}
	V_n &= 1+
	\frac{(2V_{n-1}-1)^2}{4V_{n-1}}
	= V_{n-1} + \frac{1}{4V_{n-1}}\\
	\Rightarrow V_n^2 &= V_{n-1}^2 + \frac{1}{16V_{n-1}^2} + 1/2
	\end{align*} 
	To get a precise expression
	for $V_n$, we add the differences of $V_i^2 - V_{i-1}^2$. 
	\begin{align}
	\label{eqn:Vn}
	\sum_{k=2}^{n} V_k^2 - V_{k-1}^2 &= \frac{n-1}{2} + 
	\sum_{k=1}^{n-1}\frac{1}{16V_{k-1}^2}\nonumber\\
	\Rightarrow V_n^2 - V_1^2  &= \frac{n-1}{2} + 
	\sum_{k=1}^{n-1}\frac{1}{16V_{k-1}^2}
	\end{align}
	Note that $V_1^2 = (R_1+1)^2 = 25/16$. This, coupled with~\eqref{eqn:Vn} shows 
	that the higher order term of $V_n$ is $\sqrt{\frac{n}{2}}$. We use $V_n \sim 
	\sqrt{\frac{n}{2}}$ to substitute $\frac{1}{8(k-1)}$ for the fractional 
	$\frac{1}{16V_{k-1}^2}$ term 
	in 
	the summation. Thus, rewriting~\eqref{eqn:Vn}, we get 
	\begin{align*}
	V_n^2 &\sim \frac{n-1}{2} + \frac{H_{n-1}}{8} + \frac{25}{16} \\
	&\sim \frac{n}{2} + \frac{\log n}{8} + O(1)\\
	\Rightarrow R_n &= (V_n-1) \sim \sqrt{\frac{n}{2}+\frac{\log n}{8}+O(1)}\\
	\Rightarrow p_n &= \frac{2R_n}{n} \sim \sqrt{\frac{2}{n}}\qquad 
	(\text{by Equation~\eqref{eqn:Rnpn}})\\
	\Rightarrow t_n &\sim 1 - \frac{1}{\sqrt{2n}} \qquad 
	(\text{by Equation~\eqref{eqn:tn}})
	\end{align*}
	
	Note that $t_n < 1$ for all $n$, i.e., the threshold is non-trivial for all $n$.
	Or equivalently, the fact that $R_n$ strictly increases with $n$ already shows
	that there is a non-zero amount of trade at every round in a PBE.
	
	\paragraph{Uniqueness.} Since we have assumed uniqueness for the $n-1$ round 
	PBE for all $\dist_{[0,x]}$, uniqueness for the $n$ round game follows from 
	simply the uniqueness of solutions to the set of 
	equations~\eqref{eqn:un},~\eqref{eqn:tn},~\eqref{eqn:rn}. Its summary, given by 
	equation~\eqref{eqn:rnOptimize} shows that $R_n$ has a quadratic dependence on 
	the optimization variable $t_n$: clearly the maximum is unique. The second 
	order conditions in~\eqref{eqn:rnOptimize} can be easily verified to show that 
	optima is indeed a maxima. 
\end{proof}

\begin{remark}
	Interestingly, although the price starts very low, at $p_n \sim
	\sqrt{\frac{2}{n}}$, the threshold starts very high at $t_n \sim
	1-\frac{1}{\sqrt{2n}}$. That is, the seller already starts with a very small
	price, and the buyer still refuses to buy for most of his values, waiting for
	the price to go down even further. 
\end{remark}

\section{Time discounted infinite horizon}
\label{sec:infiniteHorizon}
We formally make the connection to bargaining with one-sided information (and 
equivalently to durable goods monopoly) in this section. 
\subsection{Bargaining and Repeated Sales}
\paragraph{Bargaining with one-sided offer.} Recall that in bargaining with
one-sided offer, a single seller repeatedly makes price offers
to a single buyer for the sale of a single unit of good, till a sale is made.
The buyer's private value $v$ for the good is drawn from distribution $\dist$
and is publicly known. In each round, the seller posts a price, and the buyer
can either accept or reject the offer the seller made. Once the buyer accepts at
a price, the game ends, and the seller's revenue is that round's price. We
consider the infinite horizon bargaining game: the number of rounds in the game
is unbounded, but the buyer and seller have a common discount rate of $1-\delta$
on their utilities, i.e., utilities in round $i$ are scaled by
$(1-\delta)^{i-1}$. Equivalently, one could think of the there being a
$1-\delta$ probability of the game ending after each round.

\begin{theorem}
	\label{thm:BargainingRepeated}
	For atomless bounded support distribution $\dist$, for every pure strategy 
	threshold PBE in the bargaining game with one-sided 
	offer for $\dist$, there 
	exists a corresponding pure strategy threshold PBE 
	in the time discounted infinite horizon partial commitment repeated sales game 
	for $\dist$, s.t., if the seller's expected revenue and buyer's expected 
	utility are respectively 
	$R, U(v)$ for type $v$ in the bargaining 
	game, they are $\frac{R}{\delta}, \frac{U(v)}{\delta}$ in the repeated sales 
	game.
\end{theorem}
\begin{proof}
	Consider a PBE $(\sigma_s^*,\sigma_b^*, \mu^*)$ for the bargaining game where 
	the $\sigma$'s are the strategies of the seller and buyer, and $\mu$ is the 
	seller's beliefs. The strategy profile $(\sigma_s^\dagger,\sigma_b^\dagger, 
	\mu^\dagger)$ is a PBE for the repeated sales game where
	\begin{enumerate}
		\item $\forall k$, if $h^k = \text{Reject}^k$, set $\bigg \{\begin{array}{ll}
		\sigma_s^\dagger(\cdot|h^k) &= \sigma_s^*(\cdot|h^k) \\
		\sigma_b^\dagger(\cdot|v,h^k) &= \sigma_b^*(\cdot|v,h^k) \qquad \forall v\\
		\mu^\dagger(\cdot|h^k) &= \mu^*(\cdot|h^k)
		\end{array}$
		
		\item $\forall k$, if $h^k \ni \text{Accept}$, set $\sigma_s^\dagger(p|h^k) = 
		1$ where $p$ is the smallest price for which the buyer accepted in history 
		$h^k$. 
		
		\item $\forall k$, if $h^k \ni \text{Accept}$, set 
		$\bigg\{ \begin{array}{ll}
		\sigma_b^\dagger(\text{Accept}|v,h^k) &= 1 \qquad \forall v \geq 
		\text{current-round-price}\\
		\sigma_b^\dagger(\text{Reject}|v,h^k) &= 1 \qquad \forall v <
		\text{current-round-price}
		\end{array}$
		
		\item $\forall k$, if $h^k \ni \text{Accept}$, set 
		$\bigg\{ \begin{array}{lll}
		\mu^\dagger(\cdot|h^k) &= \mu^*(\cdot|h^{k_1+1}) & \text{if } h^k = 
		\text{Reject}^{k_1}\text{Accept}^{k_2} \text{ for } k_1 \geq 0, k_2 > 0\\
		\mu^\dagger(\high|h^k) &= 1 & \text{if } h^k = 
		\text{Reject}^{k_1}\text{Accept}^{k_2}\text{Reject}^{1}\{\text{Reject,Accept}\}^*
		
		\text{ for } k_1\geq 0, k_2 > 0 
		\end{array}$
	\end{enumerate}
	As described above, the PBE $(\sigma_s^\dagger,\sigma_b^\dagger, 
	\mu^\dagger)$ mimics the PBE $(\sigma_s^*,\sigma_b^*, 
	\mu^*)$ of the bargaining game as long as 
	the buyer has never purchased in history. If the buyer has purchased at least 
	once in the past, the bargaining game ends immediately, but the repeated sales 
	game doesn't. In this case, in the repeated sales game:
	\begin{enumerate}
		\item The seller is bound not to increase price beyond any
		price at which the buyer accepted in the past, and the seller has no reason to
		decrease the price either: so $\sigma_s^\dagger$ posts exactly the smallest
		price accepted in the past, i.e., the ``price remains fixed upon acceptance by
		the buyer'' just like the discussion in the proof of
		Theorem~\ref{thm:partialCommitmentPBEExistence}.
		\item The buyer is guaranteed to face prices lower than what he has once
		accepted, and therefore accepts any price. The condition ``$v <$
		current-round-price'' can never occur in equilibrium-path if buyer has accepted 
		at 
		least once in the past: as that would mean that
		in the past the buyer accepted in a round $r$ when $v <$ round-$r$-price.
		\item As far as beliefs are concerned, the seller never expects to see a buyer 
		who has once accepted in the past to ever reject in the future. If the seller 
		observes this, it is a zero-probability action, and beliefs are (as they can 
		be) arbitrarily updated to a point-mass at the largest point in the support 
		$\high$ (i.e., the belief is supported at $\high$ with probability $1$). If the 
		seller observes no 
		such anomaly, the seller's beliefs are simply borrowed from bargaining game. 
	\end{enumerate}
	The proof for $(\sigma_s^\dagger,\sigma_b^\dagger, 
	\mu^\dagger)$ being a PBE in the repeated sales game follows from 
	$(\sigma_s^*,\sigma_b^*, \mu^*)$ being a 
	PBE in the bargaining game. Let $U_{B,s}, U_{B,b}$ be the seller's and buyer's 
	utility in the bargaining game, and let $U_{R,s}, U_{R,b}$ be the same in the 
	repeated sales game. From the description above, it follows that:
	\begin{enumerate}
		\item $U_{R,s}(\sigma_s^\dagger,\sigma_b^\dagger, 
		\mu^\dagger) = U_{B,s}(\sigma_s^*,\sigma_b^*, 
		\mu^*)/\delta$
		\item $U_{R,b}(\sigma_s^\dagger,\sigma_b^\dagger, 
		\mu^\dagger) = U_{B,b}(\sigma_s^*,\sigma_b^*, 
		\mu^*)/\delta$
	\end{enumerate}
	I.e., the buyer, whenever he makes a single purchase at price $p$ in the 
	bargaining game, he makes infinite purchases at the same price, giving the 
	seller a revenue of $\sum_{r=0}^{\infty}p\cdot(1-\delta)^r = p/\delta$. 
	Similarly, the buyer's utility is also scaled by $\delta$. 
	
	Given that the utility structure of the two games are identical (scaling 
	doesn't 
	change the utility structure), it is immediate that 
	$(\sigma_s^\dagger,\sigma_b^\dagger, 
	\mu^\dagger)$ is a PBE for the repeated sales game.
\end{proof}
\begin{remark}[Durable goods monopoly]
	An identical connection as in Theorem~\ref{thm:BargainingRepeated} exists 
	between durable goods monopoly (\citet{Coase1972}) and the partial commitment 
	repeated sales game, given that the math for bargaining and durable goods 
	monopoly is identical.  
\end{remark}

\begin{remark}
	Note that the partial commitment repeated sales game in the finite horizon 
	setting, discussed in Section~\ref{sec:coupons}, cannot be reduced to the 
	bargaining game because, the utility structures are not identical: the number 
	of repeated purchases of an item depends on which round it was first 
	purchased. On the other hand in the infinite horizon setting, the ``number'' of 
	repeated purchases always appears as a $\frac{1}{\delta}$ factor facilitating 
	the reduction to the bargaining setting. 
\end{remark}

\subsection{The linear demand case: $U[0,1]$ distribution}
\label{sec:linearDemand}
For the rest of this section, we focus on the linear demand case, namely, the 
distribution of buyer's values is $U[0,1]$. 

\paragraph{Scale-invariant strategies for seller.} A seller's strategy is 
scale-invariant, if given two different prior beliefs 
that are scaled versions of one another, the seller's prices given these two 
beliefs are also scaled versions of one another, with the same 
ratio as prior beliefs. Concretely, for the uniform distribution, this means 
that when the belief is $U[0,k]$, there is some constant $\gamma$ independent 
of $k$ such that for all $k$, the seller's price is $\gamma k$. 

\paragraph{Stationary threshold strategies for buyer.} A buyer follows 
a stationary threshold strategy if, given a price $p$, buyers with $v 
\geq \lambda p$ purchase, and the rest reject, where $\lambda$ is independent 
of $p$ or the round in which it was offered. 

\subsubsection{Equilibrium selection and Coase conjecture}
\label{sec:eq_select}
As discussed in related work (Section~\ref{sec:related}), prior work
(\citep{ST83,Stokey81,Gul1986}) confirmed the Coase conjecture in the infinite
horizon bargaining/durable goods monopoly setting, when we focus on PBEs where
the buyer follows stationary threshold strategies, and the seller follows 
scale-invariant
strategies. By the correspondence between bargaining and partial commitment
repeated sales that we established in Theorem~\ref{thm:BargainingRepeated}, the
PBEs in these models have an equivalent in the partial commitment game, with the
seller's revenue being a factor $\frac{1}{\delta}$ larger than in bargaining.
For the sake of completeness, we work out that PBE here, adapted to our repeated
sales setting. We follow notation much along the lines of Chapter 10
in~\cite{FT91}.

\begin{theorem}[\citet{ST83}, \citet{Stokey81}]
	\label{thm:eq_select}
	In the time-discounted infinite horizon partial commitment repeated sales game 
	for the $U[0,1]$ 
	distribution, there exists a PBE in which the seller follows a scale-invariant 
	strategy, and the buyer follows a stationary threshold 
	strategy. As $\delta \to 0$, the seller's revenue approaches 
	$\frac{1}{2\sqrt{\delta}}$. 
\end{theorem}
\begin{proof}
	We look for a PBE where:
	\begin{enumerate}
		\item The seller follows a \emph{scale-invariant strategy}: i.e., when his 
		beliefs are in $U[0,t]$, he posts a price of $\gamma t$ where $\gamma$ is 
		independent of $t$. 
		\item The buyer follows 
		a \emph{stationary threshold strategy}: if offered a price $p$, buyers with $v 
		\geq 
		\lambda p$ purchase, and the rest reject, where $\lambda$ is independent of $p$ 
		or the round in which it was offered. 
	\end{enumerate}
	Let $R_{\delta}(t)$ denote the seller's discounted future revenue, given that 
	his current beliefs are in $U[0,t]$, but the distribution has a total mass of 
	just $t$ instead of $1$. Since the seller maximizes his revenue we 
	have the following (note that once the buyer purchases at $p$ the future price 
	remains fixed at $p$, just like in 
	Theorems~\ref{thm:partialCommitmentPBEExistence} and~\ref{thm:recursion}).
	\begin{align}
	\label{eqn:FT1}
	R_\delta(t) = \max_p\{(t-\lambda p)\frac{p}{\delta} + 
	(1-\delta)R_\delta(\lambda p)\}
	\end{align}
	First order conditions applied to~\eqref{eqn:FT1}, differentiating w.r.t. $p$, 
	yields
	\begin{align}
	\label{eqn:FT2}
	\frac{1}{\delta}(t-2\lambda p) + (1-\delta)\lambda R'_\delta(\lambda p) = 0
	\end{align}
	The envelope theorem, when applied to~\eqref{eqn:FT1} yields
	\begin{align}
	\label{eqn:FT3}
	R'_\delta(t) = \frac{p(t)}{\delta} = \frac{\gamma t}{\delta}
	\end{align}
	The indifference of threshold buyer yields
	\begin{align}
	\label{eqn:FT4}
	\frac{\lambda p - p}{\delta} = (1-\delta)\frac{\lambda p - \gamma\lambda 
		p}{\delta}
	\end{align}
	Combining~\eqref{eqn:FT2},~\eqref{eqn:FT3},~\eqref{eqn:FT4} yields
	\begin{align*}
	\lambda &= \frac{1}{\sqrt{\delta}}\\
	\gamma &= \frac{\sqrt{\delta} - \delta}{1-\delta}
	\end{align*}
	The seller's revenue $R_\delta(t)$, is obtained by integrating~\eqref{eqn:FT3}, 
	namely $R_\delta(t) = \frac{\gamma t^2}{2\delta}$, thus the revenue for 
	$U[0,1]$ distribution is $\frac{\gamma}{2\delta}$. As $\delta \to 0$, the 
	revenue approaches $\frac{1}{2\sqrt{\delta}}$. 
\end{proof}
\begin{remark}
	Note that this PBE indeed satisfies Coase's conjecture. As offers are made in 
	very quick succession, namely when the discount factor $1-\delta$ approaches 
	$1$, the seller's first round price, $\gamma$, approaches $0$. The monopolist's 
	revenue approaches $0$ in the bargaining game ($\sqrt{\delta}$, to be precise), 
	and in the repeated sales game it 
	approaches $\frac{1}{\sqrt{\delta}}$. The latter is much smaller than the 
	Myerson optimal static revenue of $\frac{1}{4\delta}$. 
\end{remark}
\subsubsection{Folk theorem}
\label{sec:folk}
While the selection criteria of scale-invariant strategies for
seller and stationary threshold policies for buyer helped to confirm Coase
conjecture, if we remove these stationarity assumptions,~\citet{AD89} showed a
folk theorem that as $\delta \to 0$, any revenue between $0$ and monopoly
revenue of $\frac{1}{4}$ can be obtained in a PBE for the infinite horizon time
discounted bargaining/durable goods monopolist game. By our
Theorem~\ref{thm:BargainingRepeated}, this immediately translates to a folk
theorem in the partial commitment repeated sales game, with any seller revenue
between $0$ and $\frac{1}{4\delta}$ possible. We state the folk theorem adapted
to our repeated sales game here.

\begin{theorem}[\citet{AD89}]
	\label{thm:folk}
	In the time discounted infinite horizon partial commitment repeated sales game 
	for the $U[0,1]$ distribution, 
	for every $\epsilon > 0$ there exists a $\delta(\epsilon) > 0$ such that for 
	any $\delta < \delta(\epsilon)$, the seller's revenue $R_\delta \in 
	[\frac{\epsilon}{\delta}, \frac{1/4-\epsilon}{\delta}]$.
\end{theorem}

\paragraph{Proof idea.} We refer the reader to~\citet{AD89} or Chapter 10
of~\citet{FT91} for a proof of this theorem. The main idea is to construct a 
reputational equilibrium. Let $\Delta$ denote the time interval between rounds, 
and let $r$ be the interest rate, so that the discount factor $1-\delta = 
e^{-r\Delta}$. Consider the real-time price path of $p_\tau = 
\frac{1}{2}e^{-\eta\tau}$ where $\tau$ is real time, and $p_\tau$ is the price 
posted by seller at $\tau$. This, price posted at round $n$ is $p_n = 
\frac{1}{2}e^{{-\eta n \Delta}}$. When $\eta$ is very close to $0$, the prices 
are extremely slowly decreasing, and most buyers will be impatient enough to 
buy very early. In particular, all buyers with value at least 
$\frac{1}{2}+\epsilon$ for small $\epsilon$ buy because $v-\frac{1}{2} = 
\sup_{\tau}e^{-r\tau}(v-\frac{1}{2}e^{-\eta\tau})$. This means that the seller 
makes almost his monopoly revenue in this equilibrium. But what's the catch? It 
is that the seller should not feel tempted at a later point in the path to 
decrease prices quicker than the announced path to capture more buyers --- if 
the seller ever deviates, the PBE constructed is such that the buyer will 
immediately switch to the ``Coase path'' 
discussed in Theorem~\ref{thm:eq_select}. The question is if the seller will 
ever deviate? In particular, given that $\eta$ is very small, leading to a very 
slowly decreasing price path, and consequently a very slow rate of sales as 
time proceeds making him want to switch to the Coase path. This would make the 
buyers want to wait till this switch to very low prices happen.  To ensure that 
this deviation does not happen, $\eta$ is fixed and $\Delta$ is taken to $0$, 
and a simple calculation is used to show that the seller doesn't have the 
incentive to deviate from the exponential price path. 

Qualitatively, this is a reputational equilibrium, because the seller is forced 
to maintain his reputation as someone who is determined not to reduce prices 
too quickly, and always follow the same rate parameter $\eta$. The moment the 
seller deviates from this path, the buyer decides that this seller is not 
strong and that the seller believes in the Coase conjecture, and hence switches 
to the ``Coase path'' with tiny seller revenue. 

\subsubsection{Exogenous restrictions on strategy space}
\label{sec:exogeneous}
For the linear demand case (i.e., $U[0,1]$ distribution), 
\citet{ST83},~\citet{Stokey81} focus on the PBEs where the seller follows 
scale-invariant strategies. Instead of searching for PBEs that 
satisfy certain properties (like we did in Section~\ref{sec:eq_select} where 
seller's strategies are scale-invariant and buyer has stationary threshold 
strategies), what if we exogenously enforce that the seller 
follows scale-invariant strategies? I.e., seller's strategy 
space is restricted to scale-invariant strategies that post a price of $pk$ 
when his belief is $U[0,k]$ for all $k$, with $p$ being independent of $k$. In 
other words, once the seller announces his initial price of $p$ for the 
$U[0,1]$ belief, he is immediately committing himself to the price for the 
entire rejection path (namely, just look at the current belief $U[0,k]$ and 
post a price of $pk$). 
With this exogenous enforcement, we now show that there is a unique PBE in 
which the seller can extract at least  $\frac{4}{3+2\sqrt{2}} \sim 69\%$ of 
Myerson optimal revenue benchmark, unlike the $O(\frac{1}{\sqrt{\delta}})$ 
predicted by the Coase-conjecture-confirming Theorem~\ref{thm:eq_select} 
(see Theorem~\ref{thm:infinitepartial} and the following discussion for this 
result). Further, the price in the first round as $\delta \to 0$ is about 
$0.586$, unlike the $\sqrt{\delta}$ as the first round price predicted by 
Theorem~\ref{thm:eq_select}.

\begin{theorem}\label{thm:infinitepartial}
	In the time-discounted infinite horizon partial commitment repeated sales game 
	for the $U[0,1]$ distribution, when the seller's strategy space is restricted 
	to scale-invariant strategies, there exists a unique\footnote{
		as in Theorem~\ref{thm:recursion}, the uniqueness refers to equilibrium path 
		and ignores irrelevant multiplicities arising out of several belief updates 
		that are allowed after zero probability actions by the buyer in off-equilibrium 
		paths. Further, we also ignore PBEs (if at all any exist) where a buyer rejects 
		a price that is smaller than the smallest point in the support of seller's 
		belief. It is straightforward to show by backward induction that indeed no such 
		PBEs exist in the finite horizon model, but it is not clear how to rule-out the 
		existence of such PBEs in infinite horizon model.
	}
	PBE in which, as $\delta 
	\to 0$, the seller's revenue approaches a
	$\frac{4}{3+2\sqrt{2}}$ fraction of the
	Myerson optimal revenue benchmark of $\frac{1}{4\delta}$.
\end{theorem}
\begin{proof}
	Let $R_\delta$ denote the expected revenue. Let $t(p)$ be the buyer's threshold 
	(at times abbreviated to just $t$) 
	when the seller posts a price of $p$. Indifference of the threshold buyer yields
	\begin{align}
	\label{eqn:tOfp}
	\frac{t-p}{\delta} &= (1-\delta)\frac{t-tp}{\delta}\nonumber\\
	t(p)&= \frac{p}{\delta+(1-\delta)p}
	\end{align}
	
	Since the seller is committed to following scale-invariant strategies, the 
	revenue $R_{\delta,[0,t]}$ for $U[0,t]$ (with a full mass of $1$) is simply 
	$t\cdot R_\delta$. 
	The expected revenue $R_\delta$ of the seller can be written as:
	\begin{align}\label{eqn:RP}
	R_\delta &= \max_p\bigg\{(1-\delta)tR_{\delta,[0,t]} +
	(1-t)\cdot p/\delta\bigg\}\nonumber\\
	&= \max_p\bigg\{(1-\delta)\cdot R_\delta t^2 +
	(1-t)\cdot p/\delta\bigg\}.
	\end{align}
	The first term of~\eqref{eqn:RP} is the expected revenue contribution from
	rejection in the first round: $t$ is the probability of rejection and the
	expected revenue upon rejection is $R_{\delta,[0,t]}$. The second term
	of~\eqref{eqn:RP} is the expected revenue contribution from acceptance in the
	first round: $1-t$ is the probability of acceptance in the first round, and the
	expected revenue upon acceptance is $p$ in every round appropriately 
	discounted. 
	
	Combining equations~\eqref{eqn:tOfp} and~\eqref{eqn:RP} gives
	\begin{align}\label{eqn:unmaxRP}
	R_\delta = \max_t\bigg\{\frac{t(1-t)}{(1-(1-\delta)t)(1-(1-\delta)t^2)}\bigg\}
	\end{align}
	
	First order condition on~\eqref{eqn:unmaxRP} gives:
	\begin{align}
	\label{eqn:fourthDegree}
	(1-\delta)^2t^4 - 2(1-\delta)^2t^3 +2(1-\delta)t^2-2t+1 = 0
	\end{align}
	Substitute $t=1-\theta$ in~\eqref{eqn:fourthDegree}. A bit of introspection 
	hints that $\theta = \Theta(\delta)$. Hence expand 
	equation~\eqref{eqn:fourthDegree} in $\theta, \delta$, 
	ignoring third 
	and higher order terms (i.e., retaining only constants, 
	$\theta,\delta,\theta^2,\delta^2,\theta\delta$ terms), and this gives that as 
	$\delta \to 0$, we have $\theta \sim
	\frac{\delta}{\sqrt{2}}$. 
	To compute the revenue approximtion, we substitute $t = 
	1-\frac{\delta}{\sqrt{2}}$ in~\eqref{eqn:unmaxRP} and take the ratio of the 
	expression for $R_\delta$ and the Myerson optimal revenue benchmark 
	of $\frac{1}{4\delta}$. As $\delta\to 0$, this ratio approaches from above 
	$\frac{4}{3+2\sqrt{2}}$, which is approximately $0.69$. The optimal 
	price $p$ approaches $\frac{\sqrt{2}}{\sqrt{2}+1}$. 
	
	\paragraph{Uniqueness.} This follows immediately from the uniqueness of the 
	optimization problem in~\eqref{eqn:unmaxRP} for $t \in [0,1]$.
\end{proof}
\paragraph{Discussion.} We remark a few properties of the PBE in 
Theorem~\ref{thm:infinitepartial}:
\begin{enumerate}
	\item The optimal first round price as $\delta \to 0$ is $p = 
	\frac{\sqrt{2}}{\sqrt{2}+1} 
	\approxeq 0.586$. This is very different from the tiny 
	$\sqrt{\delta}$ price found in Theorem~\ref{thm:eq_select}. So is the revenue 
	approximation factor of $0.69$, as compared to the tiny $\Theta(\sqrt{\delta}) 
	$ approximation in Theorem~\ref{thm:eq_select}. 
	\item Equation~\eqref{eqn:tOfp} says that for all $p < 1$, we have $t(p) < 1$, 
	i.e., however high the price $p$ in the first round is, there is always a small 
	fraction of buyer population that wants to buy very early. This again is very 
	different from the result in Theorem~\ref{thm:eq_select} which, for a first 
	round price of $p$ has the first round threshold $\lambda p = 
	\frac{1}{\sqrt{\delta}} p $, indicating that any first round price $p > 
	\sqrt{\delta}$ is 
	rejected by \emph{all} buyers. 
\end{enumerate}

\section{Directions for further research}
\label{sec:conc}
The basic posted-price-for-a-single-buyer setup considered in this paper is a 
special case of the $n$ buyers setting where one could run auctions with 
(potentially personalized) reserve prices, as is the common practice with most 
ad exchanges in the market for display/banner ads. Suppose there are 
$n$ buyers 
with values drawn from independent but not necessarily identical distributions, 
and suppose the seller sells an item in each round via an auction with 
(personalized) reserve prices. Given the prevalence of this auction format, 
understanding the equilibrium structure and revenue in this case is an 
interesting direction for further research.

\section*{Acknowledgements}
We are grateful to Tanmoy Chakraborty and Yashodhan Kanoria for
helpful discussions on the problem. We thank Anna Karlin for reading
earlier drafts of this paper and offering useful feedback.

\bibliographystyle{plainnat}
\bibliography{sequential}
\appendix
\section{A gentle introduction to Perfect Bayesian Equilibrium}
\label{app:intro_details}
\paragraph{The fishmonger's problem.\protect\footnote{We thank Amos Fiat for
		suggesting this name for the problem.}} There is a single seller of fish and a
single buyer who enjoys consuming a fresh fish every day. The buyer has a private
value $v$ for each day's fish, drawn from a publicly known distribution.
However, this value is drawn {\em only once}, i.e., the buyer has the same
unknown value on all days.  Each day, the seller sets a price for that day's
fish, which of course can depend on what happened on previous days. The buyer
can then decide whether to buy a fish at that price or to reject. The goal  of
the buyer is to maximize his total utility (his value minus price on each day
he buys and 0 on other days), and the goal of the seller is to maximize
profit.  How much money can the seller make in $n$ days in equilibrium?

Consider, for example, the case where the distribution of the buyer's value is
uniform in $[0,1]$ (denoted by $U[0,1]$ for short), and the game lasts for one day.
In this case, it is easy to see that the optimal seller price is the monopoly
price\footnote{This is a special case of~\cite{M81}'s theorem which implies that
	the revenue optimal mechanism for a seller facing a buyer with value
	drawn from known distribution $F$ is to offer a price of $p$ (monopoly price) that
	maximizes $p(1-F(p))$.} of 1/2, resulting in an expected seller profit of
1/4.

What prices should the seller set if the game is to last for two days?  A first
guess is 1/2 on both days, for an expected profit of 1/4 each day or 1/2
overall. But this is implausible: if the buyer rejects on the first day, the
seller might reasonably assume that the buyer's value is $U[0,1/2]$, in which
case the seller's best response is to offer a price of 1/4 on the second day.
This yields the seller strategy shown in Figure~\ref{fig:PBE1}.  However, this
buyer/seller strategy pair is {\em not} in equilibrium.  This seller strategy
is based on the fallacious assumption that the buyer's best response is to buy
on both days if his value is above 1/2. Indeed, a buyer with value $1/2 +
\epsilon$ gets a utility of $2\epsilon$ for buying both days, whereas his
utility is $1/4 + \epsilon$ if he only buys on the second day. Interestingly,
if the buyer could be {\em guaranteed} that the price on the second day was
1/2, then his best response would be to buy both days when $v > 1/2$.
However, since the seller is {\em unable to commit} to a second day price, the
buyer's strategy on the first day must take into account that on the second day
the seller will best respond to the buyer's first day strategy.  The result, in
this case, is that the buyer is incentivized to wait for the lower second day
price unless his value is at least 3/4.

\begin{figure}[h]
	\centering
	\begin{subfigure}[b]{.4\textwidth}
		\includegraphics[width=0.7\textwidth]{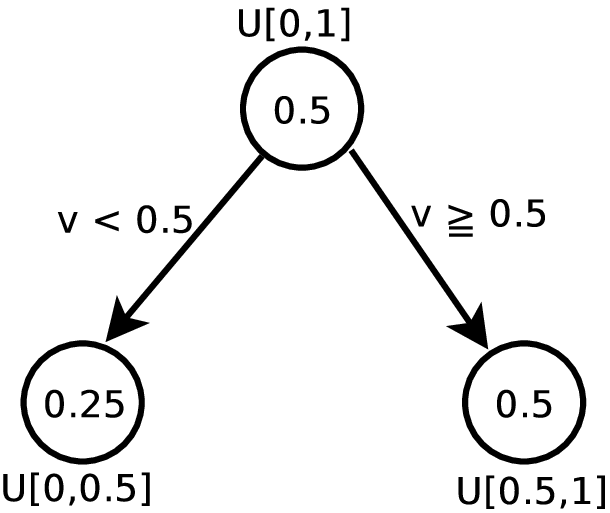}
		\caption{This is {\em not} a (perfect Bayesian) equilibrium\qquad\qquad\qquad\qquad\qquad\qquad}
		\label{fig:PBE1}
	\end{subfigure}
	\qquad\qquad
	\begin{subfigure}[b]{0.4\textwidth}
		\includegraphics[width=0.7\textwidth]{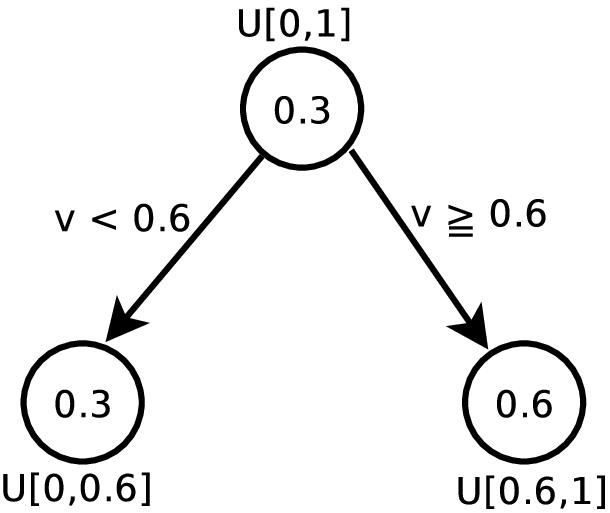}
		\caption{This is a valid (perfect Bayesian)
			equilibrium\qquad\qquad\qquad\qquad}
		\label{fig:PBE2}
	\end{subfigure}
	\caption{Equilibrium illustration for the $2$ days fishmonger's problem. The
		number in the top circle is the price on the first day. The number in the
		circle following the left arrow denotes the price on the second day after the
		buyer rejected on the first day, and the one in the circle following the right
		arrow denotes the second day price after buyer buys at the posted price on the
		first day. The distributions are updated depending on whether the buyer
		purchased or rejected on the first day.}
	\label{fig:PBE}
\end{figure}

So what is the optimal strategy for this 2-day game (for arbitrary
distributions of buyer valuation)? Or for the $n$-day version? In this paper we
study this question: how much can the seller make in $n$-days in a {\em Perfect
	Bayesian Equilibrium (PBE)}?  

\paragraph{Commitment and Perfect Bayesian Equilibrium.} The lack of commitment
in repeated games is the major driver of fundamental differences in outcomes
when compared to single-shot games. The absence of commitment in repeated games
is captured by the notion of a Perfect Bayesian equilibrium (PBE).  Informally,
a PBE consists of a seller strategy, describing what price he offers as a
function of the history of play at each time, and a (possibly randomized) buyer
strategy describing his accept/reject decisions given the history of play and
his value. For every possible value the buyer has, and for every possible
history of play, his strategy must be a best response to the {\em subtree} of
prices the seller's strategy specifies for that particular history of play. For the seller,
for every possible history of play, the subtree of prices offered henceforth
must optimize his profit given the buyer's strategy and the induced
distribution of values the buyer has (as determined by the history of play).
For example, Figure~\ref{fig:PBE2} shows the essentially unique PBE strategies
for the buyer and the seller in the example discussed above. We refer the 
reader to Section~\ref{sec:zeroCommitment} for a general way to compute PBEs 
for two round
games with arbitrary distributions. For the sake of intuition, we also flesh
out in Appendix~\ref{app:twoRounds} the $U[0,1]$ case and provide a complete 
description of the PBE strategies. 
\section{Broader related work}
\label{app:broader_related_work}
\paragraph{Ratchet effect.}
A well known fact about many of these repeated sales settings is the ``ratchet 
effect'', that the
revelation principle fails to hold \citep{ FreixasGT1985,LaffontT1988}.  

\paragraph{Posting prices.} Most of the literature assumes that the seller is
{\em restricted} to posting a price, which was justified by
\citet{Skreta2006,Skreta2015} who showed (respectively for a single buyer and
many buyers) that posting prices is optimal among all mechanisms. Both these 
papers considered the case where the seller has a single unit. 

\paragraph{Other related work in bargaining.} While our closest connection in 
this work is to bargaining with one-sided offer 
and one-sided incomplete information, there is also work bargaining with 
two-sided incomplete information and two-sided offers (\citet{FudenbergT1983}). 
We refer 
the reader to~\citet{AusubelCD2002} for a survey on bargaining with incomplete 
information. On the topic of bargaining with complete information, 
~\citet{Stahl72} and~\citet{Rubinstein1982} were the first to capture the 
intrinsically dynamic nature of bargaining and model it as a sequential 
complete information game (as an alternative to the axiomatic approach to 
Bargaining \citep{NashBargaining}). They showed that the process of sequential 
bargaining
yields a unique, Pareto-efficient outcome. While the combination of uniqueness 
and efficient outcome is remarkable and important, we note that bargaining 
gains its main 
interest from incomplete information.
\newcommand{\Agents}{\mathcal{A}} 
\newcommand{\Outcomes}{\mathcal{O}}
\newcommand{\Types}{\mathcal{T}}  
\newcommand{\Reals}{\mathbb{R}} 
\newcommand{\type}{\theta}
\newcommand{\obj}{\mbox{\sc obj}}
\newcommand{\thetas}{\vec{\theta}}
\section{Lack of commitment can never help even in very general settings}
\label{app:Myerson_optimal}
We formally define a very general model of mechanism design here and prove
Proposition~\ref{thm:MyeOPT}.
\begin{definition} {\textbf{General model of mechanism design}.}
	An instance of a mechanism design problem is given by a set of $m$ {\em agents}
	$\Agents$, a set of {\em outcomes} $\Outcomes$ and a {\em type space} $\Types$
	for each agent, where each type $\type$ is a function from  $\Outcomes$ to
	$\Reals$ (which is the {\em utility} of the agent with type $\theta$ for the
	given outcome).  Each agent $a$ has a type $\type_a$, which is her private
	information.  In the Bayesian setting, additionally, we are given a joint
	probability distribution $\dist$ over the types of all the agents,
	$\Types^m$, from which the type vector $\vec{\theta}$ of the agents is sampled.
	A {\em mechanism}  is a multi-party protocol in which the agents participate,
	as a result of which there is an outcome $o(\thetas)$.  The mechanism
	designer's goal is to maximize his {\em objective} $\Ex[\thetas\sim
	F]{\obj(\thetas)}$, where $\obj(\cdot)$ is a function from  $\Outcomes$ to
	$\Reals$. 
\end{definition} 

Note that the above definition includes, in addition to the usual cases of
welfare/revenue maximization, constraints such as budget constraints and
scenarios such as mechanism design without money, non-linear objectives such as
makespan minimization in scheduling and max-min fairness. In order to extend
this model to the repeated setting, we need to additionally specify how many
times the setting is repeated.  We allow the number of repetitions to be a
random variable.

\newcommand{\Naturals}{\mathbb{N}}
\begin{definition} {\textbf{General model of repeated mechanism design.}}
	\label{def:GRMD} 
	An instance of a repeated mechanism design problem is given by an instance of
	the mechanism design problem, the probabilities $\{ q_t\}_{t \in \Naturals}$
	with which the $t$-th repetition is realized, the fractions $\{d_t\}_{t \in \Naturals}$ 
	with which the mechanism designer and the agents discount their $t$-th round utilities.
	We require that $\sum_{t=1}^{\infty} q_t d_t < \infty$ (i.e., the process either doesn't
	continue infinitely, or if it does, agents discount their future utilities
	enough to avoid infinite utilities).  The buyer types remain the same in every
	repetition and there are no inter-round constraints except this. The repeated
	mechanism is now a protocol, which in sequence produces an outcome  $o_t$ for
	each time\footnote{The mechanism could be randomized and its outcome on day
		$t$ could depend (apart from $\thetas$) on the realization of the random coin
		tosses on days $1$ to $t-1$. To avoid excessively cumbersome notation we avoid
		spelling this out formally. But our argument and results directly extend to
		these settings too.} $t$ till the process stops (the mechanism designer and the
	agents know the probabilities $q_t$ that determine this stopping time, but get
	to know the precise stopping time only when it happens.) The utility of agent
	$a$ is the sum  $\Ex[\thetas\sim F|\theta_a]{\sum_{t=1}^\infty q_t d_{t}
		\theta_a(o_t(\thetas))} $.  The objective of the mechanism designer is
	$\Ex[\thetas\sim F]{\sum_{t=1}^\infty q_t d_{t} \obj(o_t(\thetas))}$. 
\end{definition} 

It is easy to see that the game defined in Definition \prettyref{def:basicgame}
is a special case of the above model: $\Agents = \{ 1 \}$, $\Outcomes =
\{\text{accept, reject}\} \times \Reals$, types of the form
$\theta((\text{accept},p))  = v - p $  for some $v \in \Reals$ and $\theta(o) =
0 $ otherwise, and objective $\obj( (\text{accept},p) ) = p$ and $\obj(o) =0$
otherwise. For the finite horizon model $q_t =1$ for $t \in \{1,2,\dots,n\}$,
and $q_t = 0$ for $t > n$, with $d_{t} = 1$ for all $t$. The time
discounted infinite horizon game can be described by setting $q_t = 1$ for all $t \in
\Naturals$ and $d_{t} = (1-\delta)^{t-1}$ for all $t \in
\Naturals$.  Note that $\sum_{t=1}^{\infty} q_td_{t} = 1/\delta < \infty$. 

We restate Proposition~\ref{thm:MyeOPT} formally here and prove it. 

\begin{oneshot}{Proposition \ref{thm:MyeOPT}} (A simple generalization of the 
	result 
	in~\cite{BB84})
	In the general model of repeated mechanism design, the optimal objective value
	obtained without any commitment is never larger than the optimal objective
	value obtained when commitment is possible. Formally let $\obj^*$ be the
	optimal expected objective value for the single round mechanism design problem. Then the
	optimal expected objective value attainable in any PBE in the repeated mechanism design problem is
	at most $\Ex[\thetas\sim F]{\sum_{t=1}^{\infty} q_td_{t}\obj^*}$. 
\end{oneshot}
\begin{proof}
	Suppose on the contrary that there was a PBE with expected objective value
	$\Ex[\thetas\sim F]{\sum_{t=1}^{\infty} q_td_{t}\obj(o_t(\thetas))} > \Ex[\thetas\sim F]{\sum_{t=1}^{\infty} q_td_{t}\obj^*}$.
	Consider the following mechanism for the single round game. All agents submit
	their types to the mechanism designer.  The designer chooses day $t$ with
	probability $\frac{q_td_{t}}{\sum_t q_td_{t}}$ and runs the said PBE till
	day $t$, and the outcome on day $t$ will be the outcome realized for the single
	round game. Agent $a$ with type $\theta_a$, upon truthful reporting of his type,
	will get an expected utility of $\Ex[\thetas\sim F|\theta_a]{\sum_{t=1}^\infty \frac{q_t d_{t}}{\sum_{t}q_td_t}
		\theta_a(o_t(\thetas))} $ which is just a scaled version of his utility in the
	PBE. Thus the agent has no incentive to deviate in the proposed mechanism 
	because that
	would mean that the said PBE was not really a PBE. For this mechanism, 
	the expected objective of the designer is 
	$\Ex[\thetas\sim F]{\sum_{t=1}^\infty \frac{q_t d_{t}}{\sum_t q_td_t} \obj(o_t(\thetas))}$. By our assumption,
	the former quantity is at least $\Ex[\thetas\sim F]{\sum_{t=1}^\infty \frac{q_t d_{t}}{\sum_t q_td_t} \obj^*} > \obj^*$. 
	This is a contradiction because in the single round game, it is not possible to
	get an expected objective value higher than $\obj^*$. 
\end{proof}
\section{Two Rounds Game}
\label{app:twoRounds}
\paragraph{Full solution to the $2$ days $U[0,1]$ repeated sales game.} 
We present here the full solution to the $2$ days $U[0,1]$ repeated sales game. 
Note that a PBE has to specify equilibrium behavior at off equilibrium 
paths as well. In the $U[0,1]$ example, although the first round price is $p_1 
= 0.3$, the seller still has to specify his second round behavior when the 
first round 
price is not $0.3$. 
\renewcommand{\algorithmcfname}{STRATEGY}
\IncMargin{1em}
\begin{algorithm}[!h]
	\SetKwData{Left}{left}\SetKwData{This}{this}\SetKwData{Up}{up}
	\SetKwFunction{Union}{Union}\SetKwFunction{FindCompress}{FindCompress}
	\SetKwInOut{Input}{Input}\SetKwInOut{Output}{Output}
	\textbf{\underline{Round-$1$ pricing:}} Set $p_1 = 0.3$\;
	\textbf{\underline{Round-$2$ pricing:}} \\
	\uIf{$p_1 \leq 0.5$}
	{
		\uIf{\text{Buyer rejects in round $1$}}
		{
			Set second round price of $p_2 = p_1$
		}
		\ElseIf{\text{Buyer accepts in round $1$}}
		{
			Set second round price of $\max(2p_1,0.5)$
		}
	}
	\ElseIf{$p_1 > 0.5$}
	{
		\uIf{\text{Buyer rejects in round $1$}}
		{
			Set second round price of $p_2 = 0.5$
		}
		\ElseIf{\text{Buyer accepts in round $1$}}
		{
			Set second round price of $1$
		}
	}
	\caption{Seller's strategy in the $2$ rounds \text{U[0,1]} game}
	\label{strat:SellerTwoRoundsUniform}
\end{algorithm}
\DecMargin{1em}
\IncMargin{1em}
\begin{algorithm}[!h]
	\SetKwData{Left}{left}\SetKwData{This}{this}\SetKwData{Up}{up}
	\SetKwFunction{Union}{Union}\SetKwFunction{FindCompress}{FindCompress}
	\SetKwInOut{Input}{Input}\SetKwInOut{Output}{Output}
	\textbf{\underline{Round-$1$ strategy:}}\\
	\uIf{$p_1 \leq 0.5$}
	{
		\uIf{$v \geq 2p_1$}
		{
			Accept
		}
		\Else
		{
			Reject
		}
	}
	\ElseIf{$p_1 > 0.5$}
	{
		Reject
	}
	\textbf{\underline{Round-$2$ strategy:}} \\
	\uIf{$v \geq p_2$}
	{
		Accept
	}
	\Else
	{
		Reject
	}
	\caption{Buyer's strategy in the $2$ rounds \text{U[0,1]} game}
	\label{strat:BuyerTwoRoundsUniform}
\end{algorithm}
\DecMargin{1em}
\end{document}